\documentclass[lettersize,journal]{IEEEtran}
\usepackage{amsmath}
\usepackage{tikz}
\usepackage{nicematrix}
\usepackage{algorithm}
\usepackage{algpseudocode}
\usepackage{graphicx}
\usepackage{subcaption}
\usepackage{amsthm}
\usepackage{float}
\usepackage{mathrsfs}
\usepackage{verbatim}
\usepackage{epstopdf}
\usepackage{amssymb}
\usepackage{amsfonts}
\usepackage{color}
\usepackage{cite}
\usepackage{cancel}
\usepackage[shortlabels]{enumitem}
\usepackage[breaklinks=true,letterpaper=true,colorlinks=false,bookmarks=false]{hyperref}
\usepackage{algorithm}
\usepackage{algpseudocode}
\usepackage{arydshln}
\usepackage[official]{eurosym}
\usepackage{comment}

\usepackage{amsmath,amssymb,amsthm,mathrsfs,amsfonts,dsfont}
\usepackage[shortlabels]{enumitem}

\DeclareMathOperator*{\argmax}{arg\,max}

\newtheorem{theorem}{Theorem}
\newtheorem{proposition}{Proposition}

\newtheorem{corollary}{Corollary}[theorem]

\newtheorem{lemma}{Lemma}

\newtheorem{definition}{Definition}

\usepackage{tikz}
\usepackage{pgfplots}
\pgfplotsset{compat = newest}
\usetikzlibrary{shapes.misc, positioning}
\usetikzlibrary{pgfplots.external}
\usetikzlibrary{external}

\usepackage{xcolor}

\usepackage[colorinlistoftodos,prependcaption,textsize=tiny]{todonotes}
\pgfplotscreateplotcyclelist{nano mark style}{%
dashed, every mark/.append style={solid, fill=gray!30}, mark=*, line width=1.1pt\\%
dashed, every mark/.append style={solid, fill=gray!30}, mark=square*, line width=1.1pt\\%
dashed, every mark/.append style={solid, fill=gray!30}, mark=triangle*, mark size=2.5pt, line width=1.1pt\\%
dashed, every mark/.append style={solid, fill=gray!30},mark=diamond*, mark size=2.5pt, line width=1.1pt\\%
dashed, every mark/.append style={solid, fill=gray!30},mark=x, mark size=3pt, line width=1.1pt\\%
}

\pgfplotscreateplotcyclelist{nano line style}{%
color = black, line width=2pt\\%
densely dashed, color = black, line width=2pt\\%
densely dotted, color = black, line width=2pt\\%
loosely dashed, color = black, line width=2pt\\%
}
\begin{document}
\title{Differentially Private Consistent Release of Counting Queries}
\author{Borzoo Rassouli$^1$ and Morteza Varasteh$^2$\\
\small{$^1$ Nokia Bell Labs, Stuttgart, Germany}\\
\small{$^2$ University of Essex, Colchester, UK}\\
{\tt\small borzoo.rassouli@nokia-bell-labs.com}, {\tt\small  m.varasteh@essex.ac.uk}
}
\maketitle
\begin{abstract}
We study the problem of releasing counting-query outputs through a stochastic mechanism that is both consistent and \((\epsilon,\delta)\)-differentially private. Consistency requires the released value to lie within the feasible range of the query, while utility is measured by the worst-case probability of error. We first derive a closed-form expression for the minimum achievable error probability and obtain an explicit optimal mechanism. By exploiting the active differential privacy constraints satisfied by this mechanism, we then characterize the entire class of optimal mechanisms via a propagation argument, identifying the structural properties shared by all optimizers.

We next extend the framework to the setting in which the privacy mechanism is cascaded with an arbitrary fixed stochastic transformation representing a predetermined portion of the communication medium between the source and the destination. We first establish necessary and sufficient conditions under which this partial fixation of the medium incurs no loss in utility. We then derive upper and lower bounds on the optimal achievable performance based on convex mixing and spectral perturbation. Finally, we apply the theory to (M)-ary phase-shift keying (PSK) transmission over an additive white Gaussian noise (AWGN) channel and show that uncoded transmission is effectively optimal in the high-privacy regime.

\end{abstract}

\section{introduction}\label{intro}

The widespread deployment of data-driven services has intensified the need for mechanisms that simultaneously provide strong privacy guarantees and preserve utility. Differential privacy (DP) has emerged as a principled framework for this purpose, offering protection against inference attacks while enabling meaningful data analysis \cite{dwork2006calibrating,dwork2014algorithmic, DP2014}. Its promise lies in providing strong, quantifiable protection for individuals while still allowing meaningful statistical or algorithmic utility. Over the past decade, differential privacy has been adopted across a broad spectrum of domains, including machine learning~\cite{abadi2016deep,bu2020deep}, signal processing~\cite{geng2016optimal,he2014differential}, and communication systems and now plays a pivotal role in the design of privacy-preserving mechanisms.

A fundamental class of statistical releases consists of counting queries, whose output is an integer-valued statistic, such as the number of users possessing a given attribute~\cite{counting, counting2}. Although such aggregate statistics may appear to reveal little about individual records, numerous de-anonymization and reconstruction attacks have demonstrated that they can nevertheless leak sensitive information~\cite{arvind}.

On another note, in many applications, it is not sufficient to release a private estimate of the count; the released value should also be \emph{consistent}~\cite{Barak}, meaning that it lies within the feasible range of counts. Consistency preserves the semantic validity of the released statistic and eliminates the need for additional post-processing to project infeasible outputs back onto the admissible set.

There is an extensive body of literature on the differentially private release of queries. In \cite{mcsherry}, the authors introduce the exponential mechanism to satisfy pure differential privacy for general query functions. The authors in \cite{Ghosh2008} show that in the context of pure differential privacy, the geometric mechanism, i.e., adding noise with geometric distribution, is universally optimal for a broad class of loss functions in discrete settings. The authors in \cite{hall} show how adding Gaussian noise preserves approximate differential privacy for releasing infinite dimensional functions. Based on \cite{Ghosh2008}, the authors in \cite{Gupte} derive the optimal noise probability distributions for minimax (risk-averse) users in the context of pure differential privacy. In \cite{Geng0}, they generalize the results of \cite{Ghosh2008} and \cite{Gupte} to real-valued queries functions, and show that the optimal query-output independent perturbation mechanism is the staircase mechanism. The same authors in \cite{Geng2018} introduce truncated Laplacian mechanisms and prove their superiority over Gaussian noise when one seeks tight bounds on amplitude or variance.

This paper studies the optimal consistent release of counting queries from a source/sender to a destination/recipient  under approximate differential privacy, which is a generalization of pure differential privacy. The privacy mechanism is represented by a stochastic matrix that takes query outputs as its input and maps them to the released counts. Utility is measured by the worst-case probability of error over all input values. The design objective is to find a mechanism that minimizes this worst-case error subject to $(\epsilon,\delta)$-differential privacy and consistency constraints. 

We consider two settings. In the first, we assume complete freedom to design the processing applied to the messages transmitted from the source to the destination, representing an idealized scenario. We provide a complete characterization of the minimum achievable error probability in closed form. We also derive a canonical optimal mechanism explicitly, which admits a simple interpretation as a mixture of the identity channel and a cyclic exponential kernel. Moreover, by exploiting the active differential privacy constraints satisfied by this canonical mechanism, we further characterize the entire class of optimal mechanisms through a propagation argument. In particular, we determine when the optimal mechanism is unique, when optimal mechanisms are full rank, and which structural properties are shared by all optimal solutions.

In the second setting, we extend the framework to practical scenarios in which  a portion of the medium from the sender to the recipient is predetermined. Consequently, we do not have complete freedom to arbitrarily design the processing applied to the messages from source to destination.
The overall channel is modeled as the product of a designable privacy mechanism and a predetermined stochastic transformation. We observe that the existence of a predetermined portion of the end-to-end communication chain, does not necessarily lead to a degradation in utility, namely the worst-case error probability. In this light, we provide sufficient and necessary conditions under which the same performance as the ideal case, i.e., the first setting, can still be achieved, particularly in the high-privacy regime.
Afterwards, using convex mixing and spectral perturbation, we derive computable upper and lower bounds on the optimal error probability in general. 


\textbf{Notation.} 
Matrices and vectors are denoted by bold capital letters (e.g., \(\mathbf{A}, \mathbf{Q}\)) and bold lower case letters (e.g., \(\mathbf{b}, \mathbf{z}\)), respectively. Unless explicitly stated otherwise, all vectors considered in this paper are column vectors. 
Sets are denoted by calligraphic capital letters (e.g., \(\mathcal{X}, \mathcal{G}\)), with the exception of the set of real numbers, denoted by \(\mathds{R}\). The cardinality of set \(\mathcal{P}\) is denoted by \(|\mathcal{P}|\).

For integers \(m \leq n\), we define the discrete interval \( [m:n] \triangleq \{m, m+1, \ldots, n\} \). For $x\in[0,1]$, $\Bar{x}\triangleq1-x$. The floor and ceiling operator are denoted by \(\lfloor\cdot\rfloor\) and \(\lceil\cdot\rceil\), respectively.

For a square matrix $\mathbf{A}_{M\times M}$, denote its minimum diagonal entry by $\gamma(\mathbf{A})$, i.e., $\gamma(\mathbf{A})\triangleq\min_{i\in[1:M]}a_{i,i}$, where $a_{i,j}$, or sometimes denoted by $[\mathbf{A}]_{i,j}$, denotes the entry of $\mathbf{A}$ on row $i$ and column $j$. It is also evident that \(\gamma(\cdot)\) is superadditive, i.e., \(\gamma(\mathbf{A}+\mathbf{B})\geq\gamma(\mathbf{A})+\gamma(\mathbf{B})\). The entrywise max norm, infinity norm, and the induced matrix $1$-norm of \(\mathbf{A}\) are denoted by \(\|\mathbf{A}\|_{\textnormal{max}}\triangleq\max_{i,j}|a_{i,j}|\), \(\|\mathbf{A}\|_{\infty}\triangleq\max_{i}\sum_j|a_{i,j}| \), \(\|\mathbf{A}\|_1\triangleq
\|\mathbf{A}^T\|_\infty\), respectively. The identity and all-one matrices are denoted by \(\mathbf{I}\) and \(\mathbf{J}\), respectively.


The positive cone generated by vectors \(\mathbf{t}_1,\ldots,\mathbf{t}_M\) is defined as  $\textnormal{cone}_+(\mathbf{t}_1,\ldots,\mathbf{t}_M)
    \triangleq
    \left\{
    \sum_{i=1}^M \lambda_i \mathbf{t}_i \;\middle|\; \lambda_i \geq 0
    \right\}$.
    
The interior of a set \(\mathcal{W}\) is denoted by \(\textnormal{int}(\mathcal{W})\). The indicator function \(\mathds{1}_{\{\textnormal{condition}\}}\) equals one when the condition holds, and zero otherwise. For a real number $x$, $(x)^+\triangleq\max\{x,0\}.$

Finally, consider a constraint of the form \(a\leq b\), where \(a\) and \(b\) are two parameters. If the constraint holds with equality, i.e., \(a=b\), then the constraint is said to be tight (or active).

\section{Preliminaries}
\subsection{Differential Privacy}

Two data sets $\mathcal{D}_1$ and $\mathcal{D}_2$ are neighbors (adjacent) if they differ in only one data entry (i.e., the data of an individual). Let $\epsilon \geq 0$ and $\delta\in[0,1)$. An algorithm/mechanism $A$, which is a random mapping in general, is said to be $(\epsilon, \delta)$-differentially private if
\begin{equation}
    \textnormal{Pr}\{{A}(\mathcal{D}_1)\in\mathcal{S}\} \leq e^\epsilon \textnormal{Pr}\{{A}(\mathcal{D}_2)\in\mathcal{S}\} + \delta,
\end{equation}
for all subsets $\mathcal{S}$ of the image of ${A}$ and all adjacent datasets $\mathcal{D}_1$ and $\mathcal{D}_2$. As a special case, when \(\delta=0\), the algorithm is said to satisfy \(\epsilon\)-differential privacy. This setting is commonly referred to as \emph{pure differential privacy}, whereas the case \(\delta>0\) is known as \emph{approximate differential privacy}.

\subsection{Counting query}
 Given a dataset \(\mathcal{D}\) consisting of records from some domain \(\mathcal{X}\), a counting query is specified by a predicate \(f:\mathcal{X}\to\{0,1\}\) that indicates whether a record satisfies a particular property. The query returns the number of records in \(D\) for which the predicate evaluates to one, namely,
\[
q_f(\mathcal{D})=\sum_{x\in \mathcal{D}} f(x).
\]
Equivalently, the query counts the number of individuals in the dataset satisfying a specified condition.

In the context of counting problems, the output of a query is an integer value, and since two neighbouring datasets will have counts that differ by at most one, the adjacency constraint in the definition of differential privacy maps to the adjacency of integer numbers. In other words, adjacent numbers must be made as indistinguishable as required depending on the choice of $\epsilon$ and $\delta$.

\section{Problem statement}\label{ideal}
For an arbitrary integer $M\geq 2$, let $\mathcal{M}\triangleq[1:M]$ denote the range of a counting query. This indexing is adopted solely for convenience; using $\{0,\ldots,M-1\}$ or any other set of $M$ consecutive integers is entirely equivalent, as it amounts only to a relabeling of the query outputs and does not affect the analysis.  

A desirable property of a privacy mechanism for counting queries is consistency. Throughout this work, we restrict our attention to mechanisms that produce consistent releases, i.e., whose output alphabet coincides with the input message set \(\mathcal{M}\).

By a \textit{mapping}\footnote{In this text, the terms mapping, mechanism, stochastic matrix, and (probability) transition matrix are used interchangeably.}, we refer to a generic $M$ by $M$ probability transition matrix $\mathbf{P}=[p_{i,j}]$, where $p_{i,j}$ denotes the probability that input $i$ is mapped to output $j$ ($i,j\in\mathcal{M}$).

The message set \(\mathcal{M}\) is to be transmitted from a sender (source) to a recipient (destination) through a privacy-preserving mapping. Accordingly, \(\mathcal{M}\) is referred to as the input message set, as it constitutes the set of messages supplied to the privacy mechanism.



Let $\mathcal{S}_{\epsilon, \delta}$ denote the set of all mappings $\mathbf{P}$ that are ($\epsilon, \delta$)-differentially private. More specifically,

\begin{align}
    \mathcal{S}_{\epsilon, \delta}\triangleq\bigg\{\mathbf{P}\ \bigg|\ \sum_{j\in \mathcal{U}} p_{k,j}\leq  e^\epsilon\sum_{j\in \mathcal{U}}p_{i,j} + \delta,&\ \  \forall~\mathcal{U} \ \subset\mathcal{M},\nonumber\\&\ \ \forall~ i,k: |k-i|=1,\bigg\}.\label{DEF}
\end{align}
It can be easily verified that $\mathcal{S}_{\epsilon, \delta}$ is a convex set. More specifically, for any $\lambda \in (0,1)$ and $\mathbf{P}, \mathbf{P}' \in \mathcal{S}_{\epsilon, \delta}$, we have $\lambda \mathbf{P} + (1-\lambda) \mathbf{P}'\in\mathcal{S}_{\epsilon, \delta}$.
 

For any input probability distribution, denoted by the probability mass vector $\mathbf{w} = (w_1,w_2,\ldots,w_M)$, the probability of error in an arbitrary  mapping $\mathbf{P}$ is given by
\begin{equation}    p_e(\mathbf{w},\mathbf{P})\triangleq 1-\sum_{i=1}^Mw_ip_{i,i}.
\end{equation}

From a transmission reliability perspective, it is desirable to minimize the probability of error, which depends on the input probability distribution $\mathbf{w}$. However, the design of a differentially private mapping is distribution-agnostic, i.e., it does not rely on the statistics of the input. Therefore, one approach is to minimize the worst-case error. In other words, we aim to find a differentially private mapping that minimizes the maximum probability of error across all possible input distributions. More formally,
\begin{align}
p_e^*&\triangleq\min_{\mathbf{P}\in\mathcal{S}_{\epsilon, \delta}}\max_{\mathbf{w}}\ p_e(\mathbf{w},\mathbf{P})\label{minmax}\\
&=1-\max_{\mathbf{P}\in\mathcal{S}_{\epsilon, \delta}}\gamma(\mathbf{P}),\label{maxmin}
\end{align}
which results from $\max_{\mathbf{w}} p_e(\mathbf{w}, \mathbf{P}) = 1 - \gamma(\mathbf{P})$, where $\gamma(\cdot)$ is defined in section \ref{intro}. It is worth noting that the optimization problem in (\ref{minmax}) could initially be formulated using the infimum and supremum in place of the minimum and maximum. However, since the objective function is continuous and the feasible set is compact, being a closed and bounded subset of a Euclidean space, the extrema are attained. Therefore, the use of the minimum and maximum in (\ref{minmax}) is justified.

The goal is to identify a $\mathbf{P}^* \in \mathcal{S}_{\epsilon, \delta}$ whose minimum diagonal entry is maximum. With this choice of $\mathbf{P}^*_{}$, it is ensured that $p_e(\mathbf{w}, \mathbf{P}^*) \leq p_e^*$ for any $\mathbf{w}$, i.e., the probability of error remains consistently below $p_e^*$ regardless of the input distribution.

The max-min problem in (\ref{maxmin}) can be written as
\begin{align}
\max_{\mathbf{P}\in\mathcal{S}_{\epsilon,\delta}} \gamma(\mathbf{P})
&= \max_{p_{i,j}} \ x \label{linprog} \\
\textnormal{s.t.}\quad
& \sum_{j\in\mathcal{U}} (p_{k,j}-e^\epsilon p_{i,j}) - \delta \le 0,
\ \forall\, |k-i|=1,\nonumber\\& \forall\, \mathcal{U}\subset\mathcal{M} \nonumber \\
& x - p_{i,i} \le 0,\quad i\in[1:M] \nonumber \\
& \sum_{j=1}^{M} p_{i,j} = 1,\quad i\in[1:M] \nonumber \\
& p_{i,j} \ge 0,\quad i,j\in[1:M] \nonumber
\end{align}

which is an optimization of a linear function with linear constraints, and hence, a linear program (LP). 

\section{Optimal solutions}



For an arbitrary integer $M\geq 2$, define the cyclic distance function $d_{M}:[1:M]\times[1:M]\to[0:\lfloor\frac{M}{2}\rfloor]$ as
\begin{equation}\label{fj}
    d_{M}(i,j)\triangleq\min\{|i-j|,M-|i-j|\}.
\end{equation}
\begin{lemma}\label{lem1}
    The matrix \([q^{d_M(i,j)}]_{i,j\in[1:M]}\) with \(q\in(0,1)\) is positive definite.
    
\end{lemma}
\begin{proof}
The proof is provided in Appendix~\ref{App1}.
\end{proof}

Define \(\alpha_M( \epsilon)\) as 
\begin{align}
    \alpha_M(\epsilon)
    &\triangleq \frac{1}{\sum_{j=1}^M e^{-\epsilon d_M(1,j)}}, \nonumber\\    
\end{align}
which equals \(\frac{1}{M}\) when \(\epsilon=0\), and otherwise,
\begin{equation}\label{ald}
    \alpha_M(\epsilon)=\frac{(e^\epsilon-1)e^{\lfloor\frac{M}{2}\rfloor\epsilon}}{e^{\lfloor\frac{M}{2}\rfloor\epsilon}(e^{\epsilon}+1)-e^{\epsilon\cdot\mathds{1}_{\{M\textnormal{ is even}\}}}-1},\ \epsilon>0.
\end{equation}

The following theorem provides a closed form solution to the LP in \eqref{linprog}.
\begin{theorem}\label{optimalmapping}
    The optimal value of the LP problem in (\ref{linprog}) is given by
    \begin{equation}\label{pestar}
        p_e^*=(1-\delta)(1-\alpha_M(\epsilon)),
    \end{equation}
which is attained by the following optimal solution
\begin{equation}\label{optimizer}
    \mathbf{P}^*=[p_{i,j}^*]=(1-\delta)\alpha_M(\epsilon)\left[e^{-\epsilon d_M(i,j)}\right]+\delta\mathbf{I},
\end{equation}
which can be interpreted as a mixture of two kernels, transmitting information with $\infty$-DP with probability $\delta$, and with $\epsilon$-DP with probability $(1-\delta)$. 

Furthermore, let $\mathcal{P}_{\mathrm{opt}}$ denote the set of all optimal solutions. The following statements hold:

\begin{enumerate}

\item $\mathbf{P}^*$ is positive definite unless $(\epsilon,\delta)=(0,0)$, in which case it is positive semidefinite.

    \item Only when \(M=2\) or \((\epsilon,\delta)=(0,0)\), the optimal solution is unique, i.e., \(|\mathcal{P}_{\mathrm{opt}}|=1\). Otherwise, the set of optimal solutions is infinite, i.e., \(|\mathcal{P}_{\mathrm{opt}}|=\infty.\) 
\item When $\delta = 0$, for \(\mathbf{P}\in\mathcal{P}_{\mathrm{opt}}\), we have
\begin{equation}\label{shared}
  p_{i,j}=p^*_{i,j},\ i\in\left[\min\{j,\underline{m}\}:\max\{j,\overline{m}\}\right]  
\end{equation}
with \(\underline{m}\triangleq\lfloor\frac{M+1}{2}\rfloor\), and \(\overline{m}\triangleq\lceil\frac{M+1}{2}\rceil\). Therefore, any optimizer \(\mathbf{P}\in\mathcal{P}_{\mathrm{opt}}\) shares at least 
\[\frac{2M^2+10M-1+(2M+1)(-1)^M}{8}\]

entries with the canonical optimizer \(\mathbf{P}^*\) in \eqref{optimizer}, including all entries between the diagonal and the middle row(s).  An illustrative representation of the shared entries is provided in Figure \ref{fig_1} for the cases \(M=4,M=5\).
\begin{figure}[t]
 \centering 
 \scalebox{0.15} 
 {\includegraphics{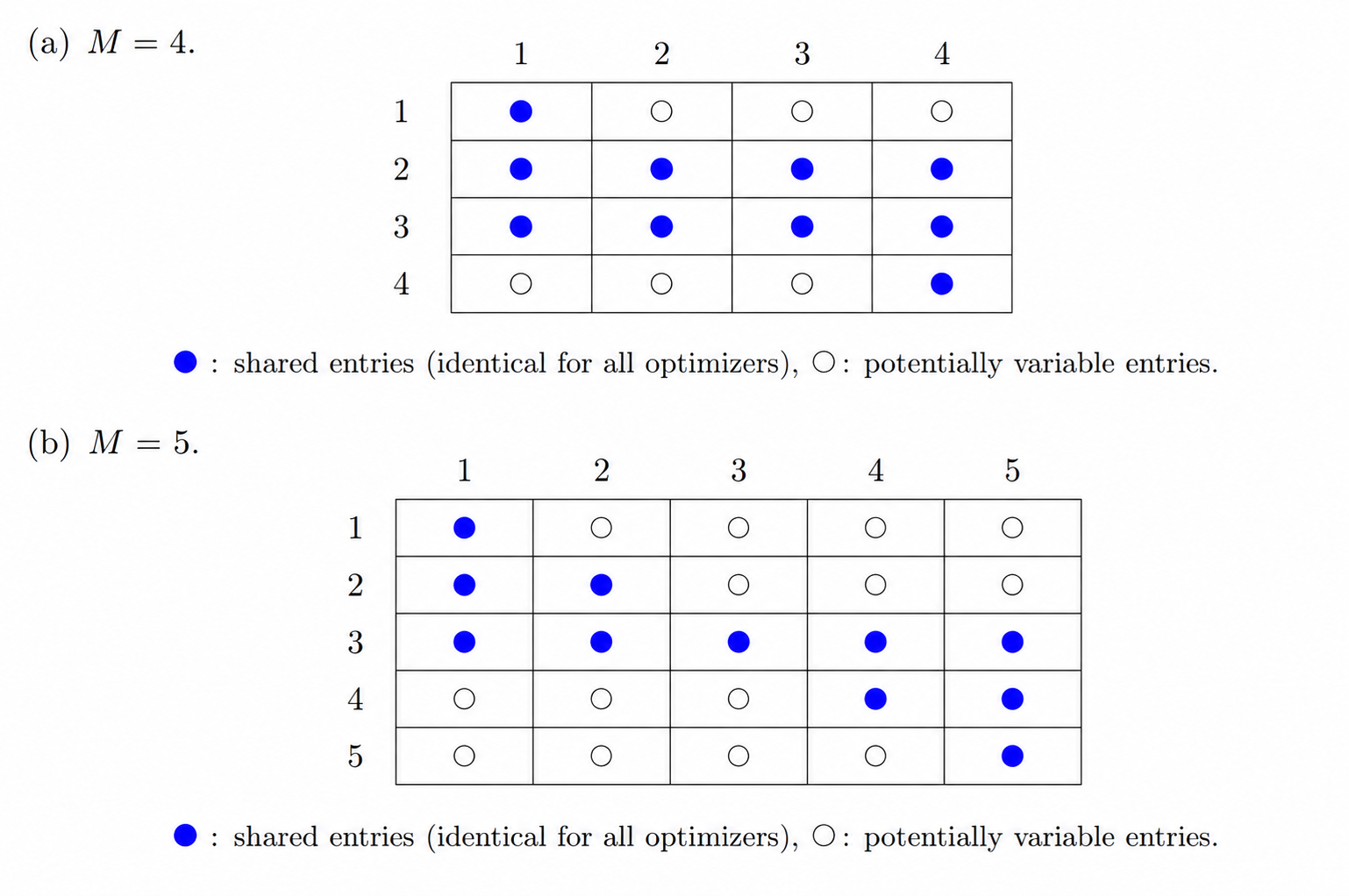}} 
 \caption{An illustrative representation of \eqref{shared} with \(\delta=0\).}
 \label{fig_1} 
\end{figure}

When $\delta>0$ all optimal solutions $\mathbf{P}\in\mathcal{P}_{\textnormal{opt}}$ share the same diagonal entries, i.e., \(p_{i,i}=p^*_{i,i}\ \forall i\). Furthermore, we have
\begin{align}
    \sum_{j=1}^ip_{i,j}&=\sum_{j=1}^ip^*_{i,j},\ \ i\leq\lceil\frac{M+1}{2}\rceil\nonumber\\
    \sum_{j=i}^Mp_{i,j}&=\sum_{j=i}^Mp^*_{i,j},\ \ i\geq\lceil\frac{M+1}{2}\rceil. \label{sumcondition}
\end{align}

\item All members of $\mathcal{P}_{\mathrm{opt}}$ are full-rank when $\delta+(1-\delta)\alpha_M(\epsilon)>\frac{1}{2}$, or when \(M=3,\ \epsilon>0\) or when \(M=4,\ \epsilon>0,\ \delta=0\).

\end{enumerate}
\end{theorem}
\begin{proof}
There are two approaches to establish \eqref{pestar}, namely a direct approach and an indirect (contradiction-based) approach. In the direct approach, we proceed as follows. The first line of constraints in the LP in \eqref{linprog} comprises \(2(M-1)2^M\) inequalities, where \(2(M-1)\) accounts for the adjacency condition and \(2^M\) corresponds to all subsets \(\mathcal{U} \subseteq \mathcal{M}\). These constraints are first equivalently reduced to \(2(M-1)\) nonlinear constraints, and subsequently reformulated as \(2(M-1)(2M+1)\) linear constraints through the introduction of auxiliary variables. The optimal solution can then be derived via strong duality and the Karush–Kuhn–Tucker (KKT) conditions. 

Alternatively, an indirect proof by contradiction may be employed, which provides a more intuitive geometric interpretation and is essential for establishing the remaining claims of the theorem. In this work, we adopt the indirect approach.

For readability, we henceforth write \(\alpha\) in place of \((1-\delta)\alpha_M(\epsilon)\), and organize the proof into several steps.



\textbf{$\mathbf{P}^*$ is feasible.} Every entry of $\mathbf{P}^*$ is non-negative, each row sums to one, where the latter can be readily verified from \(\sum_{j=1}^Me^{-\epsilon d_M(i,j)}=\frac{1}{\alpha_M(\epsilon)},\ \forall i\). As a result, $\mathbf{P}^*$ is a valid stochastic matrix. 

For any index set $\mathcal{U}\subset[1:M]$ and adjacent rows $i,k$, we can write
\begin{align}
    \sum_{j\in \mathcal{U}} p^*_{k,j}&=\alpha\sum_{j\in \mathcal{U}} e^{-\epsilon d_M(k,j)}+\delta\cdot\mathds{1}_{\{k \in \mathcal{U}\}}\label{dav}\\
    &\leq \alpha e^{\epsilon|k-i|}\sum_{j\in \mathcal{U}} e^{-\epsilon d_M(i,j)}+\delta\cdot\mathds{1}_{\{k \in \mathcal{U}\}}\label{siz}\\
    &= \alpha e^{\epsilon}\sum_{j\in \mathcal{U}} e^{-\epsilon d_M(i,j)}+\delta\cdot\mathds{1}_{\{k \in \mathcal{U}\}}\label{sizz}\\
    &= e^{\epsilon}\left(\sum_{j\in \mathcal{U}} p^*_{i,j}-\delta\cdot\mathds{1}_{\{i \in \mathcal{U}\}}\right)+\delta\cdot\mathds{1}_{\{k \in \mathcal{U}\}}\nonumber\\
    &\leq  e^{\epsilon}\sum_{j\in \mathcal{U}}p^*_{i,j} + \delta,\nonumber
\end{align}
where \eqref{dav} follows from \eqref{optimizer}, \eqref{siz} follows from the definition of \(d_M\), which is a \(1-\)Lipschitz function, and \eqref{sizz} follows from the adjacency of \(i\) and \(k\), namely, \(|i-k|=1\). Therefore, we have
$\mathbf{P}^*\in\mathcal{S}_{\epsilon, \delta}$, i.e., $\mathbf{P}^*$ is $(\epsilon,\delta)$-differentially private.

\textbf{$\mathbf{P}^*$ is an optimal solution.}
We show, by contradiction, that no other member of $\mathcal{S}_{\epsilon, \delta}$ can outperform $\mathbf{P}^*$, i.e., result in a probability of error lower than $p_e^*$ in (\ref{pestar}). The proof proceeds by a propagation of tight (or active) DP constraints. More specifically, we show that any feasible perturbation that increases all diagonal entries necessarily induces a cascade of increases in neighboring entries. These increases propagate towards the center of the matrix and force the sum of the middle row(s) to exceed one, contradicting stochasticity.

To convey the main idea, we first consider the illustrative case \(M=5\); the argument for arbitrary \(M\) is just a simple generalization and is provided in Appendix \ref{App2}. Accordingly, let \(M=5\). It is convenient to rewrite \(\mathbf{P}^*\) as

\begin{equation}\label{rewrite}
   \mathbf{P}^*=\begin{bNiceMatrix} 
   {\color{red}a}&{\color{red}b}&c&c&b\\
   {\color{red}b}&{\color{red}a}&{\color{red}b}&c&c\\
   c&{\color{red}b}&{\color{red}a}&{\color{red}b}&c\\
   c&c&{\color{red}b}&{\color{red}a}&{\color{red}b}\\
   b&c&c&{\color{red}b}&{\color{red}a}
   \CodeAfter
  \tikz \draw[dashed] (2-1.north west) rectangle (3-2.south east);
   \tikz \draw[dashed] (3-4.north west) rectangle (4-5.south east);   
    \end{bNiceMatrix},
\end{equation}
with $c\triangleq\alpha e^{-2\epsilon},\ b\triangleq e^\epsilon c,\  a\triangleq e^\epsilon b+\delta$, and the rationale for highlighting certain entries in red will become clear shortly.



Assume, for the sake of contradiction, that there exists a matrix \(\mathbf{P}'\in\mathcal{S}_{\epsilon,\delta}\) that outperforms \(\mathbf{P}^*\). In particular, since the objective in (\ref{linprog}) is to maximize the minimum diagonal entry, all diagonal entries of \(\mathbf{P}'\) must be strictly larger than \(a\).

Since \(\mathcal{S}_{\epsilon,\delta}\) is convex, it contains every point on the line segment connecting \(\mathbf{P}^*\) and \(\mathbf{P}'\). Consider moving from \(\mathbf{P}^*\) toward \(\mathbf{P}'\) along the direction \(\mathbf{P}'-\mathbf{P}^*\). Referring to (\ref{rewrite}), all diagonal entries, i.e., \(a\)'s, must increase. Moreover, to maintain the \((\epsilon,\delta)\)-DP constraints, the first superdiagonal and subdiagonal entries, i.e., \(b\)'s, that satisfy the active constraint
$a=e^{\epsilon}b+\delta
$
must also increase. Indeed, increasing \(a\) while keeping the corresponding \(b\) unchanged would violate the privacy constraint. For this reason, the entries that must increase when moving away from \(\mathbf{P}^*\) are highlighted in red in (\ref{rewrite}). Therefore, in \(\mathbf{P}'\), all \(13\) highlighted entries are strictly larger than their corresponding values in \(\mathbf{P}^*\), namely \(a\) and \(b\).

Focusing on the third row in (\ref{rewrite}), observe that the second, third, and fourth entries must increase. We show that the sum of the entries in this row necessarily exceeds one. This contradicts the requirement that each row of a stochastic matrix sums to one.

Focusing on the left dashed rectangle in \ (\ref{rewrite}), consider the sum of the first two entries of row \(2\), namely \(b+a\), and compare it with the corresponding sum in row \(3\), namely \(c+b\). We observe that
$b+a = e^\epsilon(c+b)+\delta,
$
which means that the \((\epsilon,\delta)\)-DP constraint associated with adjacent rows \(2\) and \(3\) and the index set \(\mathcal{U}=\{1,2\}\), i.e., 
\begin{align*}
b+a=\sum_{j\in\mathcal{U}}p^*_{2,j}
\le e^\epsilon \sum_{j\in\mathcal{U}}p^*_{3,j}+\delta
= e^\epsilon(c+b)+\delta,
\end{align*}
is active, i.e., equality holds. As a result, any increase in \(b+a\), which is the case here, necessarily induces an increase in \(c+b\) in order to preserve the privacy constraint. For brevity, we denote this implication by $(b+a)\uparrow \Longrightarrow (c+b)\uparrow$. A similar argument applies to the right dashed rectangle in \(\mathbf{P}^*\) in (\ref{rewrite}). By comparing the sum of the last two entries of row \(4\) with the corresponding sum in row \(3\), we conclude that any increase in the former necessarily induces an increase in the latter in order to preserve the privacy constraint. Consequently, the sums of both the first two and the last two entries of row \(3\) must increase. Furthermore, since the third entry of row (3) already increases, it follows that the sum of the entries in row \(3\) must exceed its current value of one, thereby violating the stochasticity.

We have therefore reached a contradiction. Hence, no matrix \(\mathbf{P}'\in\mathcal{S}_{\epsilon,\delta}\) can have all diagonal entries strictly larger than those of \(\mathbf{P}^*\) while remaining a valid probability transition matrix. It follows that \(\mathbf{P}^*\) is optimal.

The extension of this reasoning to arbitrary \(M\) does not yield any novel insights and is provided in Appendix~\ref{App2}.

\textbf{Proof of statement 1.}
In Lemma \ref{lem1}, replace \(q\) with \(e^{-\epsilon}\) ($\epsilon>0$) and note that a convex combination of two positive definite matrices is positive definite. At \((\epsilon,\delta)=(0,0)\), we have \(\mathbf{P}^*=\frac1M\mathbf{J}\), which is positive semidefinite.




\textbf{Proof of statement 2.}
First, we establish that there exist infinitely many optimal solutions whenever $M \geq 3$ and $(\epsilon,\delta) \neq (0,0)$. To this end, we construct a matrix $\mathbf{P}' \in \mathcal{P}_{\textnormal{opt}}$, distinct from $\mathbf{P}^*$. Consequently, every convex combination
\(
\lambda \mathbf{P}^* + (1-\lambda)\mathbf{P}', \ \lambda \in [0,1],
\)
is feasible and achieves the same objective value. It follows that the optimization problem admits infinitely many optimal solutions.


For a sufficiently small $\zeta>0$, construct a new matrix $\mathbf{P}'$ by replacing $p^*_{1,2}$ and $p^*_{1,M}$ in \(\mathbf{P}^*\) with
\[
p'_{1,2}\triangleq p^*_{1,2}+\zeta
\quad\text{and}\quad
p'_{1,M}\triangleq p^*_{1,M}-\zeta,
\]
respectively, while leaving all other entries unchanged. Clearly, $\mathbf{P}' \neq \mathbf{P}^*$, and $\mathbf{P}'$ remains a stochastic matrix with the same diagonal entries as $\mathbf{P}^*$.
It remains to verify that $\mathbf{P}'$ satisfies the differential privacy constraints. 

If $\delta=0$, from the condition $(\epsilon,\delta)\neq(0,0)$, we must have $\epsilon>0$. Therefore, the DP constraint in this case is equivalent to the requirement that \(p'_{1,2}\in[e^{-\epsilon} p_{2,2}^*,e^\epsilon p^*_{2,2}]\) and \(p'_{1,M}\in[e^{-\epsilon} p_{2,M}^*,e^\epsilon p^*_{2,M}]\). Since $e^\epsilon>1$, for sufficiently small $\zeta>0$, these conditions are satisfied.

If $\delta>0$, in the proof of the optimality of $\mathbf{P}^*$, we observed that the differential privacy constraints are active only when one of the compared probabilities is a diagonal entry. Therefore, it suffices to verify the constraints corresponding to the subsets
\[
\mathcal{U}=\{2\},\ \{1,2\},\ \{1,M\}.
\]
The set \(\mathcal{U}=\{2,M\}\) is excluded since \(p'_{1,2}+p'_{1,M}\) remains unchanged.

A straightforward calculation shows that since \(\delta> 0\), for sufficiently small $\zeta>0$, all of the constraints corresponding to these index sets remain satisfied. Hence,
\(
\mathbf{P}' \in \mathcal{S}_{\epsilon,\delta}.
\)


Finally, when $M=2$ or $(\epsilon,\delta)=(0,0)$, we observe that, as in the proof of the optimality of $\mathbf{P}^*$, any perturbation of $\mathbf{P}^*$ necessarily degrades the utility. Hence, no other feasible matrix can attain the same optimal objective value, and therefore $\mathbf{P}^*$ is the unique optimal solution.

\textbf{Proof of statement 3.}
We next prove that when $M\geq 3$ and $\delta=0$, every optimal solution has the same diagonal entries, the same middle row(s), and the same entries between the middle row(s) and the diagonal.

The proof follows the same underlying idea as the proof of the optimality of $\mathbf{P}^*$, namely, the propagation of the tightness of the differential privacy constraints. 


Assume that $M$ is odd, and there is an optimal solution $\mathbf{P}'$ whose row $\frac{M+1}{2}$ differs from that of $\mathbf{P}^*$. Since the row sums to one, we must have $p'_{\frac{M+1}{2},k}< p^*_{\frac{M+1}{2},k}$ for some $k\in[1:M]$. By $\epsilon$-DP constraint, we must have
\begin{align*}
    p'_{k,k}&\leq (e^\epsilon)^{|\frac{M+1}{2}-k|}p'_{\frac{M+1}{2},k}\\
    &<(e^\epsilon)^{|\frac{M+1}{2}-k|}p^*_{\frac{M+1}{2},k}\label{as2}\\
    &=(e^\epsilon)^{|\frac{M+1}{2}-k|}(1-p_e^*)e^{-|\frac{M+1}{2}-k|\epsilon}\nonumber\\
    &=(1-p_e^*)\nonumber.
\end{align*}
Therefore, $\mathbf{P}'$ is not optimal since it has a diagonal entry ($p'_{k,k}$) strictly lower than $1-p_e^*$. 

For even $M$, the procedure follows similarly, and hence, is omitted. As a result, we conclude that all the optimizers have the same row $\frac{M+1}{2}$ (for odd $M$), and rows $\frac{M}{2}$ and $\frac{M}{2}+1$ (for even $M$). 

To establish the common diagonal entries, assume that $\mathbf{P}'$ is an optimizer with $p'_{k,k}> p^*_{k,k}$ for some $k\in[1:M]$. However, this implies $p'_{\frac{M+1}{2},k}>p^*_{\frac{M+1}{2},k}$ (for odd $M$) and $p'_{\frac{M}{2},k}>p^*_{\frac{M}{2},k}$ (for even $M$), which, as shown earlier, contradicts the optimality of \(\mathbf{P}'\). Therefore, all optimal solutions must also share the same diagonal entries as \(\mathbf{P}^*\) when $\epsilon>0$ and $\delta=0$. 

Since the \(\epsilon\)-DP constraint is active for the entries between the middle row(s) and the diagonal entries in $\mathbf{P}^*$, the same propagation argument shows that these entries in every \(\mathbf{P}\in\mathcal{P}_{\textnormal{opt}}\) also coincide with those of $\mathbf{P}^*$.

Finally, when $\delta>0$, the proof that all optimal solutions have the same diagonal entries is stronger than the corresponding argument used to establish the optimality of $\mathbf{P}^*$. In the latter, it suffices to rule out the existence of an optimal solution whose diagonal entries are all strictly larger than those of $\mathbf{P}^*$. Here, however, we prove the stronger statement that no diagonal entry can exceed its counterpart in $\mathbf{P}^*$. Despite this stronger conclusion, this proof and that of \eqref{sumcondition} do not require additional ideas and are therefore deferred to Appendix~\ref{App2}.

\textbf{Proof of statement 4.} If \(\delta+(1-\delta)\alpha_M(\epsilon)>\frac{1}{2}\), all optimal solutions are strictly diagonally dominant; that is, the diagonal entry in each row is strictly larger than the sum of the off-diagonal entries in that row. This can be readily verified by noting that (i) each diagonal entry in any optimal solution is \(\delta + (1-\delta)\alpha_M(\epsilon)\), and (ii) each row sums to one. Therefore, all solutions are strictly diagonally dominant, and by the Gershgorin circle theorem~\cite{varga}, or by Levy–Desplanques theorem \cite{horn2012matrix}, they are full rank.  

When \(M=3,\ \epsilon>0\), defining $k\triangleq e^\epsilon>1$ and $\delta'\triangleq\frac{k+2}{1-\delta}\delta$, all the optimal solutions \(\mathbf{P}\in\mathcal{P}_{\textnormal{opt}}\) have the following form
\begin{equation}\label{QQ}
    \mathbf{P}\triangleq \frac{1-\delta}{k+2}\begin{bmatrix}
    k+\delta' & {u'} & {u} \\
    1 & k+\delta' & 1\\
    v & v' & k+\delta' \\  
    \end{bmatrix},
\end{equation}
where the diagonal entries and the entries of rows \(2\) coincide with those of \(\mathbf{P}^*\) in \eqref{optimizer}, as established in this theorem. Furthermore, the four degrees of freedom satisfy two equality constraints, namely, $u+u' = v+v'= 2$, to ensure stochasticity of $\mathbf{P}$, and eight inequality constraints, namely, ${1}\leq u',v'\leq k(k+\delta'),\ \frac{1}{k}(1-\delta')^+\leq u,v\leq k$, to ensure \((\epsilon,\delta)\)-DP. 

Obviously, when $\delta'\geq 1$, the matrix becomes strictly diagonally dominant (since $k+\delta'>2$), and hence full rank. Therefore, in what follows, we consider the case $\delta'\in[0,1)$. The matrix \(\mathbf{P}\) can be equivalently written as
\[
\mathbf{P}=\frac{1-\delta}{k+2}
\begin{bmatrix}
k+\delta' & 2-u & u\\
1 & k+\delta' & 1\\
v & 2-v & k+\delta'
\end{bmatrix},
\ 
u,v\in\left[\frac{1-\delta'}{k},\,1\right],
\]
where the equality constraints only affect the upper bounds on \(u,v\).

A direct computation of the determinant yields
\[
\det(A)
=(\frac{1-\delta}{k+2})^3f(u,v),
\]
where
\begin{equation*}
    f(u,v)\triangleq (k+\delta')^3-(k+\delta')\bigl(uv+4-(u+v)\bigr)-2\bigl(uv-(u+v)\bigr).
\end{equation*}
In Appendix \ref{App4}, we show that \(f(u,v)>0\). 
Consequently, when $M=3,\epsilon>0$, any optimizer \(\mathbf{P}\in\mathcal{P}_{\textnormal{opt}}\) is full rank.

When \(\epsilon>0, \delta=0, M=4\), all the optimal solutions have the following form
\begin{equation}\label{QQ}
    \mathbf{P}\triangleq \frac{1}{(k+1)^2}\begin{bmatrix}
    k^2 & {x_1} & {x_2} &{ x_3}\\
    k & k^2 & k & 1\\
    1 & k & k^2 & k\\
    {y_3} & { y_2} & { y_1} & k^2    
    \end{bmatrix},
\end{equation}
where the diagonal entries and the entries of rows \(2\) and \(3\) coincide with those of \(\mathbf{P}^*\) in \eqref{optimizer}, as established in this theorem. Furthermore, the six degrees of freedom satisfy two equality constraints, namely, $x_1+x_2+x_3 = 2k+1,\ 
    y_1+y_2+y_3 = 2k +1$, to ensure stochasticity of $\mathbf{P}$, and twelve inequality constraints, namely, ${k}\leq x_1,y_1\leq k^3,\ 
    {1}\leq x_2,y_2 \leq k^2,\ \frac{1}{k}\leq x_3,y_3\leq k$, to ensure \(\epsilon\)-DP. 

Calculation of the determinant in this case is not as straightforward as the case with \(M=3\) even with the extra assumption that $\delta=0$. Instead, we follow a different approach.

Define
\begin{equation*}
    \mathbf{B}\triangleq \begin{bmatrix}
    k & -1 & 0 & 0\\
    0 & k & -1 & 0\\
    0 & -1 & k & 0\\
    0 & 0 & -1 & k  
    \end{bmatrix},
\end{equation*}
which is full rank, since $k>1$. If \(\mathbf{B}\mathbf{P}\) is full rank, it follows that \(\mathbf{P}\) is also full rank. To establish this, we have
\begin{equation*}
    \mathbf{BP}=\frac{1}{(k+1)^2}\begin{bmatrix}
    k^3-k & kx_1-k^2 & kx_2-k & kx_3-1\\
    k^2-1 & k^3-k & 0 & 0\\
    0 & 0 & k^3-k & k^2-1\\
    ky_3-1 & ky_2-k & ky_1-k^2 & k^3-k  
    \end{bmatrix}.
\end{equation*}
From the aforementioned constraints on \(x_i\) and \(y_i\), \(i\in[1:3]\), all entries of \(\mathbf{B}\mathbf{A}\) are nonnegative. Moreover, we have

\begin{align}
    kx_1-k^2 + kx_2-k + kx_3-1&=k(x_1+x_2+x_3)\nonumber\\
    &\ \ \ -k^2-k-1\nonumber\\
    &=k(2k+1)-k^2-k-1\label{cons1}\\
    &=k^2-1\nonumber\\
    &<k^3-k\label{cons2},
\end{align}
where \eqref{cons1} follows from the constraint \(x_1+x_2+x_3=2k+1\), and \eqref{cons2} follows from \(k>1\). Similarly, we can replace $x_i$'s with $y_i$'s $i\in[1:3]$. Therefore, the diagonal entry in each row is strictly larger than the sum of the remaining entries in that row. Hence, \(\mathbf{BP}\) is strictly diagonally dominant and is full rank by the Gershgorin circle theorem~\cite{varga}, or by Levy–Desplanques theorem \cite{horn2012matrix}. It follows that any optimizer \(\mathbf{P}\in\mathcal{P}_{\textnormal{opt}}\) is full rank.

\end{proof}

\begin{corollary}
    $p_e^*$ increases with $M$ for fixed $\epsilon, \delta$, and decreases with $\epsilon, \delta$ for fixed $M$. Furthermore, we have
    \begin{align*}
        \lim_{M\to\infty}p_e^*=\frac{2(1-\delta)}{1+e^{\epsilon}}\ ,\ \lim_{\epsilon, \delta\to 0}p_e^*
        \ =\frac{M-1}{M}.
    \end{align*}
\end{corollary}
\begin{proof}
    The proof is straightforward and omitted for brevity.
\end{proof}
A claim that all optimizers are full-rank is tempting but is too strong, as shown by the following counterexample.
\textbf{Example.} With $M=5,\ \epsilon =\ln1.1,\ \delta = 0$, an optimal solution is given by
\begin{equation*}
    \mathbf{P}=\frac{1}{5.41}\begin{bmatrix}
    1.21 & 1.331 & 1 & .9091 & .9599\\
    1.1 & 1.21 & 1.1 & 1 & 1\\
    1 & 1.1 & 1.21 & 1.1 & 1\\
    .9524 & 1.0476 & 1.1 & 1.21 & 1.1\\
    1 & 1.1 & 1 & 1.1 & 1.21
    \end{bmatrix},
\end{equation*}
whose second column is $1.1$ times its first column. Therefore, $\mathbf{P}$ is not full rank.

We conclude this section by presenting a property of the optimal set $\mathcal{P}_{\mathrm{opt}}$, which will be used in the next section.
\begin{definition}
The radius of the optimal set $\mathcal{P}_{\mathrm{opt}}$ around \(\mathbf{P^*}\) is defined as
\[
\Delta_M(\epsilon,\delta)
\triangleq
\sup_{\mathbf{P}\in\mathcal{P}_{\mathrm{opt}}}
\|\mathbf{P}-\mathbf{P}^*\|_{\max}.
\]
\end{definition}

Let the maximum and minimum off-diagonal entries of \(\mathbf{P}^*\) be denoted by
\begin{align*}
    \Phi(\mathbf{P}^*)&\triangleq(1-\delta)\alpha_M(\epsilon)e^{-\epsilon}\\
    \phi(\mathbf{P}^*)&\triangleq(1-\delta)\alpha_M(\epsilon)e^{-\lfloor\frac M2\rfloor\epsilon}.
\end{align*}

\begin{lemma}\label{lem2}
An upper bound on the radius of $\mathcal{P}_{\mathrm{opt}}$ centered at \(\mathbf{P^*}\) is provided as follows.
\begin{equation}
    \Delta_M(\epsilon,\delta)
\le
R_M(\epsilon,\delta),
\end{equation}
with \(R_M(\epsilon,\delta)\) defined as
\begin{align}
   R_M(\epsilon,\delta)&\triangleq\max\{\Phi_0-\phi(\mathbf{P}^*),\Phi(\mathbf{P}^*)-\phi_0\}, \label{RM}
\end{align}
where
\begin{align*}
   \Phi_0&\triangleq\min\{\overline{\Phi},(1-\alpha_M(\epsilon))(1-\delta)\} \\
   \overline{\Phi}&\triangleq e^{(M-1)\epsilon}(\delta+(1-\delta)\alpha_M(\epsilon))+\delta\frac{1-e^{(M-1)\epsilon}}{1-e^\epsilon}\\
   \phi_0&\triangleq\max\{\underline{\phi},0\}\\
   \underline{\phi}&\triangleq\frac{\delta+(1-\delta)\alpha_M(\epsilon)-\delta\frac{1-e^{(M-1)\epsilon}}{1-e^\epsilon}}{e^{(M-1)\epsilon}}.      
\end{align*}
\end{lemma}
\begin{proof}
    The proof is provided in Appendix \ref{App3}.
\end{proof}

\section{Practical considerations}
In many practical scenarios, a portion of the end-to-end communication chain is predetermined by the underlying application. Consequently, one generally does not have complete freedom to arbitrarily design the processing applied to the messages from source to destination. For example, when transmitting a set of messages over an additive white Gaussian noise (AWGN) channel using a given \(M\)-ary modulation scheme at a fixed signal-to-noise ratio SNR, the end-to-end transition can be represented by the product \(\mathbf{P}\mathbf{T}\), where \(\mathbf{T}\) denotes the stochastic matrix imposed by the application (i.e., the modulation scheme and the AWGN channel), while \(\mathbf{P}\) represents the mapping that can be designed to satisfy desired requirements.

Therefore, for a given communication medium characterized by the stochastic matrix $\mathbf{T}$, the problem becomes that of designing \(\mathbf{P}\) such that the overall transition matrix \(\mathbf{P}\mathbf{T}\) minimizes the worst-case error probability while satisfying \((\epsilon,\delta)\)-differential privacy. Let \(p_e^*(\mathbf{T})\) denote the corresponding optimal value, and let \(\mathbf{P}_{\mathbf{T}}^*\) denote an optimal solution. Evidently, when \(\mathbf{T}\) is the identity matrix, \(p_e^*(\mathbf{T})\) and \(\mathbf{P}_{\mathbf{T}}^*\) reduce to \(p_e^*\) and \(\mathbf{P}^*\), respectively, as obtained in the previous section.

It is readily verified that \(p_e^*(\mathbf{T})\) and \(\mathbf{P}_{\mathbf{T}}^*\) can be obtained by solving the linear program in \eqref{linprog} with \(p_{i,j}\) replaced by \([\mathbf{P}\mathbf{T}]_{i,j}\). This procedure yields the exact optimum, although solving the LP may become computationally demanding for large values of $M$.

However, solving the LP can often be avoided. Indeed, one may first verify certain conditions which, when satisfied, allow the optimal performance to be characterized without solving the LP. Even when these conditions do not hold, useful performance bounds can still be derived from structural properties of \(\mathbf{T}\). Such bounds may be sufficient in applications where an exact solution is not required. The remainder of this section is devoted to this analysis.

We first note that
\begin{equation}\label{LB1}
    p_e^*(\mathbf{T})\geq p_e^*,
\end{equation}
which follows immediately from 
\begin{equation*}
    \max_{\substack{\mathbf{P}:\ \mathbf{PT}\in\mathcal{S}_{\epsilon,\delta}}}\gamma(\mathbf{PT})\leq \max_{\substack{\mathbf{P}:\ \mathbf{P}\in\mathcal{S}_{\epsilon,\delta}}}\gamma(\mathbf{P}),
\end{equation*}
since for any matrix $\mathbf{P}$ feasible for the left-hand side, the product $\mathbf{PT}$ is feasible for the right-hand side.

As a result, if there exists a stochastic matrix $\mathbf{P}$ such that $\mathbf{PT}$ coincides with $\mathbf{P}^*$ in \eqref{optimizer}, then $\mathbf{P}$ is an optimal solution and we have $p_e^*(\mathbf{T})=p_e^*$. 

Therefore, before resorting to the LP for an exact solution, it is natural to first determine whether such a matrix \(\mathbf{P}\) exists.\footnote{It is important to note that failure to find such a $\mathbf{P}$ does not imply that $p_e^*(\mathbf{T})\neq p_e^*$, since the solution of the LP in \eqref{linprog} is not unique in general.} 
Obviously, if $\mathbf{T}$ is invertible and $\mathbf{P}^*\mathbf{T}^{-1}\geq 0$ entrywise, then $\mathbf{P}^*_\mathbf{T}\triangleq\mathbf{P}^*\mathbf{T}^{-1}$ is a valid stochastic matrix. Indeed, its rows sum to one, since 
$\mathbf{P}^*\mathbf{T}^{-1}\mathbf{1}=\mathbf{P}^*\mathbf{T}^{-1}(\mathbf{T1})=\mathbf{1}$. Consequently, $p_e^*(\mathbf{T})=p_e^*$.

An interesting observation is that the presence of \(\mathbf{T}\) in the model, i.e., the predetermined portion of the end-to-end communication chain, does not necessarily lead to a degradation in utility, namely the worst-case error probability. This motivates the study of sufficient and/or necessary conditions on \(\mathbf{T}\) under which the utility is preserved, i.e., $p_e^*(\mathbf{T})=p_e^*$. The following theorem provides an answer to this question.


\begin{theorem}
    For an invertible $\mathbf{T}$, the following statements hold.
    \begin{enumerate}
    \item If
\begin{equation}\label{nec}
\min_{1\le i,j\le M}
[\mathbf{P}^*\mathbf{T}^{-1}]_{ij}
<
-R_M(\epsilon,\delta)\,
\|\mathbf{T}^{-1}\|_{1},
\end{equation}
then, 
\begin{equation}
    p_e^*(\mathbf{T})>p_e^*,
\end{equation}
where $R_M(\epsilon,\delta)$ is given in \eqref{RM}.
    \item If either of the following conditions, i.e., 
    \begin{align}
        \|\mathbf{T}^{-1}-\mathbf{I}\|_{\infty}\leq(1-\delta)\alpha_M(\epsilon)e^{-\lfloor\frac{M}{2}\rfloor\epsilon},
    \end{align}
    or
     \begin{align}
        \|\mathbf{T}^{-1}-\mathbf{I}\|_{\textnormal{max}}\leq\frac{(1-\delta)\alpha_M(\epsilon)e^{-\lfloor\frac{M}{2}\rfloor\epsilon}}{M[(1-\delta)\alpha_M(\epsilon)+\delta]},
    \end{align}
    holds, we have
    \begin{equation*}
        p_e^*(\mathbf{T})=p_e^*.
    \end{equation*}
        \item  If $\mathbf{1}^T\in\textnormal{int}(\textnormal{cone}_+(\mathbf{t}_1,\ldots,\mathbf{t}_M))$, there exist $\epsilon_0,\delta_0>0$ such that
    \begin{equation}\label{tasavi}
        p_e^*(\mathbf{T})=p_e^*,\ \ \forall \epsilon,\delta: \epsilon\leq\epsilon_0,\ \delta\leq\delta_0.
    \end{equation}
    \item If $\mathbf{1}^T\not\in\textnormal{cone}_+(\mathbf{t}_1,\ldots,\mathbf{t}_M)$, 
no $\epsilon_0,\delta_0>0$ exist such that \eqref{tasavi} holds.
\item If $\mathbf{T}$ is circulant, when \(\delta=0\), we have $p_e^*(\mathbf{T})=p_e^*$ if and only if $\mathbf{P}^*\mathbf{T}^{-1}\geq 0$ entrywise.
\item If $\mathbf{T}$ is singular, and $(1-\delta)\alpha_M(\epsilon)+\delta>\frac{1}{2}$ or $M=3,\epsilon>0$ or \(M=4,\ \epsilon>0,\ \delta=0\), we have $p_e^*(\mathbf{T})>p_e^*$.
    \end{enumerate}
\end{theorem}
\begin{proof}
For the first statement, let $\mathbf{P}\in\mathcal{P}_{\mathrm{opt}}$ and define \(\mathbf{E}\triangleq\mathbf{P}-\mathbf{P}^*\).
Then
\[
\mathbf{P}\mathbf{T}^{-1}
=
\mathbf{P}^*\mathbf{T}^{-1}
+
\mathbf{E}\mathbf{T}^{-1}.
\]
Moreover,
\begin{align*}
  |[\mathbf{E}\mathbf{T}^{-1}]_{ij}| &=|\sum_k e_{i,k}[\mathbf{T}^{-1}]_{k,j}|\nonumber\\
  &\leq \sum_k |e_{i,k}||[\mathbf{T}^{-1}]_{k,j}|\label{trtr}\\
  &\leq \|\mathbf{E}\|_{\textnormal{max}}\sum_k |[\mathbf{T}^{-1}]_{k,j}|\nonumber\\
  &\leq \|\mathbf{E}\|_{\textnormal{max}}\max_j\sum_k |[\mathbf{T}^{-1}]_{k,j}|\nonumber\\
  &\leq \|\mathbf{E}\|_{\textnormal{max}} \|\mathbf{T}^{-1}\|_1\nonumber\\
  &=\Delta_M(\epsilon,\delta)
\|\mathbf{T}^{-1}\|_{1}\nonumber\\
&\leq R_M(\epsilon,\delta)
\|\mathbf{T}^{-1}\|_{1}.
\end{align*}
Therefore, if \eqref{nec} holds, for any \(\mathbf{P}\in\mathcal{P}_{\mathrm{opt}}\), we have
\[
\min_{i,j}
[\mathbf{P}\mathbf{T}^{-1}]_{ij}
\le
\min_{i,j}
[\mathbf{P}^*\mathbf{T}^{-1}]_{ij}
+
\Delta_M(\epsilon,\delta)\,
\|\mathbf{T}^{-1}\|_{1}
<0,
\]
which proves the claim.

For the second statement, we note that when \(\mathbf{T}\) is invertible, we have
\begin{align*}
    \mathbf{P^*}\mathbf{T}^{-1}&=\mathbf{P^*}(\mathbf{I}+\mathbf{T}^{-1}-\mathbf{I})\\
    &=\mathbf{P^*}+\mathbf{P^*}(\mathbf{T}^{-1}-\mathbf{I}).
\end{align*}
As a result, a sufficient condition for \(\mathbf{P^*}\mathbf{T}^{-1}\geq\mathbf{0}\), is the requirement that the maximum entry of \(\mathbf{P^*}(\mathbf{T}^{-1}-\mathbf{I})\) is no larger than the minimum entry of \(\mathbf{P}^*\), i.e.,
\begin{equation*}
    \|\mathbf{P^*}(\mathbf{T}^{-1}-\mathbf{I})\|_{\textnormal{max}}\leq (1-\delta)\alpha_M(\epsilon)e^{-\lfloor\frac{M}{2}\rfloor\epsilon}.
\end{equation*}
Using the upper bounds on the max entry norm, we know
\begin{align*}
    &\|\mathbf{P^*}(\mathbf{T}^{-1}-\mathbf{I})\|_{\textnormal{max}}\leq\\&\min\left\{M\|\mathbf{P^*}\|_{\textnormal{max}}\|\mathbf{T}^{-1}-\mathbf{I}\|_{\textnormal{max}}\ ,\ \|\mathbf{P^*}\|_{\infty}\|\mathbf{T}^{-1}-\mathbf{I}\|_{\infty}\right\}
\end{align*}
Finally by noting that \(\ \|\mathbf{P^*}\|_{\infty}=1\) (due to stochasticity) and \(\ \|\mathbf{P^*}\|_{\textnormal{max}}=(1-\delta)\alpha_M(\epsilon)+\delta\), the sufficient conditions are derived.

    For the third statement, we proceed as follows. If the all-one vector lies in the interior of \(\textnormal{cone}_+(\mathbf{t}_1,\ldots,\mathbf{t}_M)\), there exist $\lambda_i>0$ such that $\mathbf{1}^T=\sum_{i=1}^M\lambda_i\mathbf{t}_i$ or equivalently $\mathbf{T}^T\mathbf{\lambda}=\mathbf{1}$. This implies that $(\mathbf{T}^T)^{-1}\mathbf{1}=(\mathbf{T}^{-1})^{T}\mathbf{1}>0$ entrywise. Equivalently, each column sum of $\mathbf{T}^{-1}$ is positive.

    In the sequel, we show that, in the high-privacy regime, i.e., for sufficiently small \(\epsilon,\delta\), \(\mathbf{P}^*\mathbf{T}^{-1}\) is entrywise nonnegative and hence constitutes a valid probability transition matrix, which results in \eqref{tasavi}.
 
 Let $s>0$ denote the minimum of the column sums of $\mathbf{T}^{-1}$. From Theorem \ref{optimalmapping}, we have $p_{i,j}^*\to\frac{1}{M},\ \forall i,j\in[1:M]$ when $\epsilon,\delta\to 0$. Hence, there exist $\epsilon_0,\delta_0>0$, such that $|p_{i,j}^*-\frac{1}{M}|<\frac{s}{M^2\|\mathbf{T}^{-1}\|_{\textnormal{max}}}, \ \forall i,j\in[1:M]$ when $\epsilon\leq\epsilon_0$ and $\delta\leq\delta_0$. As a result, we can write
 $\mathbf{P}^*=\frac{1}{M}\mathbf{J}+\mathbf{K}$ with $\mathbf{J}$ denoting the all-one matrix and $|k_{i,j}|<\frac{s}{M^2\|\mathbf{T}^{-1}\|_{\textnormal{max}}}, \ \forall i,j\in[1:M]$ when $\epsilon\leq\epsilon_0,\ \delta\leq\delta_0$. 
 
 Therefore, $\mathbf{P}^*\mathbf{T}^{-1}=\frac{1}{M}\mathbf{J}\mathbf{T}^{-1} +\mathbf{K}\mathbf{T}^{-1}$ when $\epsilon\leq\epsilon_0,\ \delta\leq\delta_0$. The minimum entry in the first term, i.e., \(\frac{1}{M}\mathbf{J}\mathbf{T}^{-1}\), is $\frac{s}{M}$. For the second term, i.e., \(\mathbf{K}\mathbf{T}^{-1}\), each entry is bounded as follows.
 \begin{align}
     |[\mathbf{K}\mathbf{T}^{-1}]_{i,j}|&=\bigg|\sum_{z=1}^Mk_{i,z}[\mathbf{T}^{-1}]_{z,j}\bigg|\nonumber\\
     &\leq \sum_{z=1}^M|k_{i,z}|\bigg|[\mathbf{T}^{-1}]_{z,j}\bigg|\label{trieq}\\
     &< \sum_{z=1}^M \frac{s}{M^2\|\mathbf{T}^{-1}\|_{\textnormal{max}}}\|\mathbf{T}^{-1}\|_{\textnormal{max}}\label{jayg}\\
     &=\frac{s}{M}\nonumber,
 \end{align}
 where \eqref{trieq} follows from the triangle inequality and \eqref{jayg} is obtained by replacing the corresponding upper bounds. Therefore, the minimum entry of \([\mathbf{K}\mathbf{T}^{-1}]\) is greater than $-\frac{s}{M}$, which makes all the entries of $\mathbf{P}^*\mathbf{T}^{-1}$ non-negative.
 Hence, in the high-privacy regime, we have $p_e^*(\mathbf{T})=p_e^*$. 

For the fourth statement, namely when \(\operatorname{cone}_+(\mathbf{t}_1,\ldots,\mathbf{t}_M)\) does not contain the all-one vector, a similar argument to that used at the beginning of the proof of the third statement shows that \(\mathbf{T}^{-1}\) must have at least one column whose entries sum to a negative value. It suffices to show that \(p_e^*(\mathbf{T})>p_e^*\) when \(\epsilon=\delta=0\). Suppose, to the contrary, that \(p_e^*(\mathbf{T})=p_e^*\). Then, we must have
\(\mathbf{P}_{\mathbf{T}}^*=\frac{1}{M}\mathbf{J}\mathbf{T}^{-1}.
\)
However, since \(\mathbf{T}^{-1}\) has a column with negative sum, \(\mathbf{P}_{\mathbf{T}}^*\) contains negative entries and therefore cannot be a valid probability transition matrix. Hence, no $\epsilon_0,\delta_0>0$ exist such that \eqref{tasavi} holds. This statement could also be proved by \eqref{nec} in the high-privacy regime, i.e., when \(\epsilon,\delta\to 0\), we have \(\mathbf{P}^*\to\frac1M\mathbf{J}\) and \(R_M(\epsilon,\delta)\to 0\).

For the fifth statement, assume that \(\mathbf{T}\) is a circulant stochastic matrix. One direction of the claim, namely the "if" part, is immediate. To prove the converse, suppose that \(p_e^*(\mathbf{T})=p_e^*\) for some \(\epsilon\). We show that \(\mathbf{P}^*\mathbf{T}^{-1}\) must be entrywise nonnegative.

Assume, to the contrary, that \([\mathbf{P}^*\mathbf{T}^{-1}]_{p,q}<0\) for some \(p,q\in[1:M]\). Since \(\mathbf{T}\) is circulant, so is its inverse \(\mathbf{T}^{-1}\). Consequently, \(\mathbf{P}^*\mathbf{T}^{-1}\) is also circulant, as circulant matrices form a commutative algebra. Therefore, the negative entry \([\mathbf{P}^*\mathbf{T}^{-1}]_{p,q}\) appears at least once in every row, and in particular, in row \(\frac{M+1}{2}\) when \(M\) is odd, and in row \(\frac{M}{2}\) when \(M\) is even.

Hence, some entry of row \(\frac{M+1}{2}\) (for odd \(M\)) or row \(\frac{M}{2}\) (for even \(M\)) of \(\mathbf{P}^*\mathbf{T}^{-1}\) is negative. In other words, the middle row of $\mathbf{P}^*$ times some column of $\mathbf{T}^{-1}$ is negative. However, Theorem~\ref{optimalmapping} establishes that all optimal solutions share these middle row(s). It follows that if \(\mathbf{P}^*\mathbf{T}^{-1}\) contains a negative entry, then so does \(\mathbf{P}'\mathbf{T}^{-1}\) for every other optimizer \(\mathbf{P}'\). Consequently, no optimizer can induce a valid probability transition matrix through multiplication by \(\mathbf{T}^{-1}\), which contradicts the assumption that \(p_e^*(\mathbf{T})=p_e^*\). Therefore, \(\mathbf{P}^*\mathbf{T}^{-1}\) must be entrywise nonnegative.

Finally, the sixth statement follows by noting that when \(\mathbf{T}\) is singular,  \(\mathbf{P}\mathbf{T}\) cannot be full rank for any stochastic \(\mathbf{P}\), and hence, cannot coincide with any optimal solution of \eqref{linprog} under the conditions of the statement 4 of Theorem 1. Therefore, \(p_e^*(\mathbf{T})>p_e^*\).
\end{proof}
\begin{corollary}
    For an invertible doubly stochastic matrix \(\mathbf{T}\), the utility can always be preserved in the high-privacy regime, which follows from the first statement of the previous theorem, since \(\mathbf{T}^{-1}\mathbf{1}=\mathbf{1}>0\). 

\end{corollary}
We now proceed to derive lower and upper bounds on \(p_e^*(\mathbf{T})\).

For a given \(\mathbf{T}\), we define \(\mathbf{Q}_{\mathbf{T}}\) as the column-wise maximizer assignment matrix induced by \(\mathbf{T}\) as follows. It is a matrix with exactly one entry equal to one and \(M-1\) zeros in each row. In row \(i \in [1:M]\), the location of the unit entry is chosen to be the row index corresponding to a maximum entry in column \(i\) of \(\mathbf{T}\); if multiple maximizers exist, one is selected arbitrarily.
More specifically,
\begin{equation}\label{QT}
    [\mathbf{Q_T}]_{i,j}\triangleq\mathds{1}_{\{j\in\argmax_kt_{k,i}\}}.
\end{equation}
With this definition, for an arbitrary stochastic matrix $\mathbf{P}$, we have
\begin{equation}\label{gua}
   [\mathbf{PT}]_{i,i}\leq[\mathbf{Q_T T}]_{i,i},\ \forall i\in\mathcal{M},
\end{equation}
since the left hand side, being a convex combination of the entries in column $i$ of $\mathbf{T}$, is at most equal to its maximum entry, i.e., the right hand side.

As an example, we have
\begin{equation*}
    \mathbf{T}=\begin{bmatrix}
        .5&.4&.1\\
        .3&.2&.5\\
        .2&.2&.6
    \end{bmatrix}, \mathbf{Q_T}=\begin{bmatrix}
        1&0&0\\
        1&0&0\\
        0&0&1
    \end{bmatrix},\mathbf{Q_T T}=\begin{bmatrix}
        .5&.4&.1\\
        .5&.4&.1\\
        .2&.2&.6
    \end{bmatrix}
\end{equation*}
Define $l({\mathbf{T}})\triangleq 1-\gamma(\mathbf{Q_T T})$. From \eqref{gua}, we get a lower bound on $p_e^*(\mathbf{T})$ as follows.
\begin{align}
    p_e^*(\mathbf{T})&=1-\gamma(\mathbf{P_T^*T})\nonumber\\
    &\geq 1-\gamma(\mathbf{Q_T T})\nonumber\\
    &=l({\mathbf{T}}).\label{LB2}
\end{align}
Note that without any privacy constraints, i.e., \((\epsilon,\delta)=(\infty,1)\), we have \(\mathbf{P_T}^*=\mathbf{Q_T}\), which may not be equal to the identity matrix.

Therefore, combining \eqref{LB1} and \eqref{LB2}, we have 
\begin{equation}\label{LB3}
    p_e^*(\mathbf{T})\geq\max\{p_e^*,l(\mathbf{T})\}.
\end{equation}
As for the upper bounds, we proceed as follows. A trivial upper bound is
\begin{equation}\label{triv}
    p_e^*(\mathbf{T}) \leq u_0(\mathbf{T}) \triangleq 1 - \gamma(\mathbf{P}^*\mathbf{T}),
\end{equation}
which follows from the postprocessing property of differential privacy: since $\mathbf{P}^*$ already satisfies the DP constraints, so does $\mathbf{P}^*\mathbf{T}$.

In the sequel, we propose three upper bounds, based on convex mixing and spectral perturbation, that outperform the trivial bound in \eqref{triv} if $\mathbf{T}$ is either invertible or has positive entries.

In the first bound, we assume that \(\mathbf{T}\) is invertible, and define
\begin{align*}
    \lambda^*&\triangleq\max\{\lambda\in(0,1] | \lambda\mathbf{P}^*\mathbf{T}^{-1}+\Bar{\lambda}\mathbf{P}^*\geq \mathbf{0}\}\\
    & = \min_{i,j: [\mathbf{P^*T}^{-1}]_{i,j}<p^*_{i,j}}\frac{p^*_{i,j}}{p^*_{i,j}-[\mathbf{P^*T}^{-1}]_{i,j}},
\end{align*}
where \(\lambda^*>0\) since the entries of \(\mathbf{P^*}\) are strictly positive.

With this \(\lambda^*\), define the first suboptimal mapping as
\begin{align}
    \mathbf{G_T^{(1)}}\triangleq \lambda^*\mathbf{P}^*\mathbf{T}^{-1}+\Bar{\lambda^*}\mathbf{P}^*\label{GT}.
\end{align}

The rationale is as follows. If \(\mathbf{T}\) is invertible and \(\mathbf{P}^*\mathbf{T}^{-1}\) is not a valid probability transition matrix, we can perturb it to construct a feasible, albeit suboptimal, solution. To this end, since all the entries of \(\mathbf{P}^*\) are strictly positive, we can choose \(\lambda\in(0,1]\) such that
$\lambda\mathbf{P}^*\mathbf{T}^{-1}+\Bar{\lambda}\mathbf{P}^*$ has non-negative entries, and therefore, is a stochastic matrix, since all its rows sum to one.

Since \(\mathcal{S}_{\epsilon,\delta}\) is a convex set, \(\mathbf{G}_{\mathbf{T}}^{(1)}\mathbf{T}\), being a convex combination of two \((\epsilon,\delta)\)-DP mappings, also satisfies the \((\epsilon,\delta)\)-DP constraints. Therefore,
\(u_1(\mathbf{T}) \triangleq 1-\gamma(\mathbf{G_T^{(1)} T})
\) is an upper bound on $p_e^*(\mathbf{T})$. Moreover, we have
\begin{align*}
  u_1(\mathbf{T})&=  1-\gamma(\mathbf{G_T^{(1)} T})\\
  &=1-\gamma (\lambda^*\mathbf{P}^*+\Bar{\lambda^*}\mathbf{P^*T})\\
  &=1-\lambda^*\gamma(\mathbf{P}^*)-\Bar{\lambda^*}\gamma(\mathbf{P^*T})\nonumber\\
  &\leq 1-\gamma(\mathbf{P}^*\mathbf{T})\\
  &=u_0(\mathbf{T}),
\end{align*}
where we used the fact that all the diagonal entries of $\mathbf{P}^*$ are equal.



We proceed to obtain a second suboptimal mapping if all the entries of \(\mathbf{T}\) are nonzero. For a given \(\mathbf{T}\), we can decompose it as \(\mathbf{T}=\mathbf{T}_1\mathbf{T}_2\), where \(\mathbf{T}_1\) is invertible with rows summing to one and \(\mathbf{T}_2\) is a valid stochastic matrix. A trivial decomposition is obtained by taking \(\mathbf{T}_1 \triangleq \mathbf{I}\) and \(\mathbf{T}_2 \triangleq \mathbf{T}\). However, a nontrivial decomposition could be constructed by choosing \(\mathbf{T}_1\) as a perturbation of the identity, namely \(\mathbf{T}_1 \triangleq \mathbf{I}+ \mathbf{X}\), or \(\mathbf{T}_1 \triangleq (\mathbf{I}+ \mathbf{X})^{-1}\), for some carefully selected $\mathbf{X}$.

With the decomposition \(\mathbf{T}=\mathbf{T}_1\mathbf{T}_2\), where both \(\mathbf{P}^*\mathbf{T}_1^{-1}\) and \(\mathbf{T}_2\) remain stochastic, we can construct the suboptimal mapping \(\mathbf{P}\triangleq\mathbf{P}^*\mathbf{T}_1^{-1}\), for which \(\mathbf{P}\mathbf{T}=\mathbf{P}^*\mathbf{T}_2\). Since \(\mathbf{P}^*\mathbf{T}_2\) is obtained from \(\mathbf{P}^*\) via post-processing, by virtue of \(\mathbf{T}_2\) being stochastic, it continues to satisfy the required differential privacy constraints.


To formalize this intuition, we set \(\mathbf{T}_1\triangleq(\mathbf{I}+\mathbf{X})^{-1}\), with the following constraints 
\begin{enumerate}
    \item \(\mathbf{I}+\mathbf{X}\) is invertible,
    \item \(\mathbf{P}^*(\mathbf{I}+\mathbf{X})\) is stochastic,
    \item \((\mathbf{I}+\mathbf{X})\mathbf{T}\) is stochastic.
\end{enumerate}
The first condition is equivalent to requiring that none of the eigenvalues of \(\mathbf{X}\) equals \(-1\).

For the second condition, we require \(\mathbf{P}^*(\mathbf{I}+\mathbf{X}) \geq 0\) entrywise and
\(\mathbf{P}^*(\mathbf{I}+\mathbf{X})\mathbf{1}=\mathbf{1}\), where the latter is equivalent to \(\mathbf{X}\mathbf{1}=\mathbf{0}\). 

Finally, the third condition is equivalent to requiring that \((\mathbf{I}+\mathbf{X})\mathbf{T} \geq 0\) entrywise, since its rows already sum to one.

We relax this problem in several steps. First, we restrict to a rank-one perturbation, i.e., we set $\mathbf{X} \triangleq \theta \mathbf{u}\mathbf{v}^\top$, where $\theta \in \mathbb{R}$ and $\mathbf{u}, \mathbf{v} \in \mathbb{R}^n$ satisfy $\|\mathbf{u}\|_2 = \|\mathbf{v}\|_2 = 1$ and $\mathbf{v}^\top \mathbf{1} = 0$.

Since \(\theta\mathbf{u_1v_1^T}\) has only one potentially nonzero eigenvalue, namely $\theta\mathbf{v}_1^T\mathbf{u}_1$, the first condition on the invertibility is equivalent to $\theta\mathbf{v}_1^T\mathbf{u}_1\neq -1$. 
Moreover, the second and third conditions are equivalent to the following entrywise inequalities 
\begin{align}
    \mathbf{P}^*+\theta(\mathbf{P}^*\mathbf{u})\mathbf{v}^T&\geq\mathbf{0}\label{heu1}\\
    \mathbf{T}+\theta\mathbf{u}(\mathbf{v}^T\mathbf{T})&\geq\mathbf{0}\label{heu2}
\end{align}

The quantity of interest that we seek to increase with a proper choice of feasible \(\theta,\mathbf{u},\mathbf{v}\) is
\begin{equation*}
    \gamma(\mathbf{P}^*\mathbf{T}_2)=\gamma(\mathbf{P}^*\mathbf{T}+\theta\mathbf{P^*uv^T T}).
\end{equation*}
In what follows, we force \(\theta\mathbf{P^*uv^T T}\) to have non-negative diagonal entries.

Since \(\mathbf{P}^*\) is invertible according to Theorem 1, we can set \(\mathbf{u}\triangleq \mathbf{P^*}^{-1}\mathbf{T}^T\mathbf{v}\). If we also impose that \(\theta\geq 0\), we get
\begin{align}
    \gamma(\mathbf{P}^*\mathbf{T}_2)&=\gamma(\mathbf{P}^*\mathbf{T}+\theta\mathbf{P^*}\mathbf{P^*}^{-1}\mathbf{T}^T\mathbf{vv^T T})\nonumber\\
    &=\gamma\left(\mathbf{P}^*\mathbf{T}+\theta\mathbf{T}^T\mathbf{v}(\mathbf{T}^T\mathbf{v})^T\right)\nonumber\\
    &\geq \gamma(\mathbf{P}^*\mathbf{T})+\theta\gamma\left(\mathbf{T}^T\mathbf{v}(\mathbf{T}^T\mathbf{v})^T\right)\label{reason}\\
    &\geq \gamma(\mathbf{P}^*\mathbf{T})\label{reason2},
\end{align}
where \eqref{reason} follows from the superadditivity of \(\gamma(\cdot)\) and \eqref{reason2} follows from two observations, namely, i) \(\theta\geq 0\) and ii) \(\gamma(\mathbf{xx^T})\geq 0\) for any \(\mathbf{x}\in\mathds{R}^M\).

Looking at \eqref{reason}, we can select $\mathbf{v}$ as the vector maximizing $\|\mathbf{T}^\top \mathbf{v}\|_2$, subject to the constraint that it is orthogonal to $\mathbf{1}$. To this end, we define
\begin{equation*}
    \mathcal{H}\triangleq\{\mathbf{x}\in\mathbb{R}^M \mid \mathbf{x}^\top\mathbf{1}=0\},
\end{equation*}
and let $\mathbf{H}$ denote the orthogonal projector onto $\mathcal{H}$, given by
\begin{equation*}
    \mathbf{H} \triangleq \mathbf{I}-\frac{1}{M}\mathbf{1}\mathbf{1}^\top.
\end{equation*}

Since $\mathbf{v}$ is required to lie in $\mathcal{H}$, it is invariant under projection, i.e., \(\mathbf{H}\mathbf{v}=\mathbf{v}\).
Therefore,
\begin{align*}
\max_{\substack{\mathbf{v}: \|\mathbf{v}\|_2=1\\
\mathbf{v}^\top \mathbf{1}=0}}
\|\mathbf{T}^\top \mathbf{v}\|_2
&=
\max_{\substack{\mathbf{v}: \|\mathbf{v}\|_2=1\\\ \mathbf{v}\in\mathcal{H}}}
\|\mathbf{T}^\top \mathbf{v}\|_2 \\
&=
\max_{\substack{\mathbf{v}: \|\mathbf{v}\|_2=1}}
\|\mathbf{T}^\top \mathbf{H}\mathbf{v}\|_2.
\end{align*}
Hence, we take $\mathbf{v}$ as the right singular vector of $\mathbf{T}^\top \mathbf{H}$ corresponding to its largest singular value. 

Let \(\mathbf{A}\triangleq -\mathbf{T^Tvv^T}\) and \(\mathbf{B}\triangleq-\mathbf{{P^*}^{-1}T^Tvv^T T}\). The entrywise inequalities in \eqref{heu1} and \eqref{heu2} reduce to $\theta\in[0,\theta_R]$, where
\begin{align*}
    \theta_R
    &\triangleq
    \min\left\{
        \min_{i,j:\,a_{i,j}>0}\frac{p^*_{i,j}}{a_{i,j}},
        \min_{i,j:\,b_{i,j}>0}\frac{t_{i,j}}{b_{i,j}}
    \right\},
\end{align*}
with the convention that a minimum over an empty set equals ($+\infty$). We note that a sufficient condition for the feasibility of a nontrivial \(\theta\), i.e., \(\theta>0\), for this scheme is that all the entries of \(\mathbf{T}\) be positive.
We set \(\theta^*\triangleq \theta_R\) if \(1+\theta_R\mathbf{v^T u}\neq 0\), otherwise, \(\theta^*\triangleq\theta_R-\zeta\) for sufficiently small \(\zeta>0\).

Finally, the second suboptimal mapping is given by
\begin{equation}
    \mathbf{G_T^{(2)}}\triangleq \mathbf{P}^*+\theta^*\mathbf{T^T vv}^T, 
\end{equation}
which results in
\begin{align*}
    p_e^*(\mathbf{T})&\leq u_2(\mathbf{T})\triangleq 1-\gamma(\mathbf{G_T^{(2)}}\mathbf{T})\\
    &\leq u_0(\mathbf{T}).
\end{align*}


We can derive a third suboptimal scheme if all the entries of \(\mathbf{T}\) are nonzero. We set 
\begin{equation}\label{const3}
    \mathbf{T}=(\beta\mathbf{P^*}^{-1}\mathbf{Q_T}+\Bar{\beta}\mathbf{I})^{-1}\mathbf{T}_2,
\end{equation}
where \(\mathbf{Q_T}\) is given in \eqref{QT}.

Denoting the eigenvalues of \(\mathbf{P^*}^{-1}\mathbf{Q_T}\) by \(\lambda_i(\mathbf{P^*}^{-1}\mathbf{Q_T})\), the constraints on \(\beta\) are
\begin{enumerate}
    \item \(\frac{\beta-1}{\beta}\neq\lambda_i(\mathbf{P^*}^{-1}\mathbf{Q_T}),\ \forall i\), for invertibility,
    \item \(\beta\mathbf{P^*}^{-1}\mathbf{Q_T T}+\Bar{\beta}\mathbf{T}\geq 0\) for stochasticity.
\end{enumerate}
Obviously, when the entries of \(\mathbf{T}\) are positive, for sufficiently small \(\beta>0\), the above constraints are satisfied.

With the decomposition in \eqref{const3}, we get
\begin{align}
    \gamma(\mathbf{P}^*\mathbf{T}_2)
    &= \gamma\!\left(\beta \mathbf{Q}_T \mathbf{T}
    + \bar{\beta}\,\mathbf{P}^* \mathbf{T}\right)\nonumber\\
    &\geq \beta \gamma(\mathbf{Q}_T \mathbf{T})
    + \bar{\beta}\,\gamma(\mathbf{P}^* \mathbf{T})
    \label{rea1}\\
    &\geq \beta \gamma(\mathbf{P}^*\mathbf{T})
    + \bar{\beta}\,\gamma(\mathbf{P}^* \mathbf{T})
    \label{rea2}\\
    &= \gamma(\mathbf{P}^* \mathbf{T}) \nonumber,
\end{align}
where \eqref{rea1} follows from the supperadditivity of \(\gamma\), and \eqref{rea2} from \eqref{gua}.

Let \(\mathbf{C}\triangleq(\mathbf{I}-\mathbf{P^*}^{-1}\mathbf{Q_T})\mathbf{T}\), and define
\begin{equation*}
    \beta_R\triangleq \min_{i,j: c_{i,j}>0}\frac{t_{i,j}}{c_{i,j}}.
\end{equation*}
The feasible set for \(\beta\) becomes 
\[
[0,\beta_R]\backslash\left\{\beta\ \bigg|\ \frac{\beta-1}{\beta}=\lambda_i(\mathbf{P^*}^{-1}\mathbf{Q_T}),\ \textnormal{for some } i\right\}.
\]
We set \(\beta^*\triangleq \beta_R\) if \(\frac{\beta_R-1}{\beta_R}\neq\lambda_i(\mathbf{P^*}^{-1}\mathbf{Q_T})\ \forall i\), otherwise, \(\beta^*\triangleq\beta_R-\zeta\) for sufficiently small \(\zeta>0\).

Finally, the third suboptimal mapping is given by
\begin{equation}
    \mathbf{G_T^{(3)}}\triangleq \beta^*\mathbf{Q_T}+\bar{\beta^*}\mathbf{P}^* , 
\end{equation}
which results in
\begin{align*}
    p_e^*(\mathbf{T})&\leq u_3(\mathbf{T})\triangleq 1-\gamma(\mathbf{G_T^{(3)}}\mathbf{T})\\
    &\leq u_0(\mathbf{T}).
\end{align*}

The following proposition summarizes the bounds derived in this section.

\begin{proposition}
    For a given $\mathbf{T}$, we have
    \begin{align}
        \max\{p_e^*,l(\mathbf{T})\}\leq p_e^*(\mathbf{T})&\leq \min\{u_1(\mathbf{T}),u_2(\mathbf{T}),u_3(\mathbf{T})\}.
    \end{align}
\end{proposition}

\section{Numerical results}
In what follows, we consider a specific scenario in which the messages are transmitted over an AWGN channel with uncoded $M$-ary phase-shift keying (MPSK) modulation.

In this context, the $M$ messages are mapped to $M$ equally spaced points on a circle of radius $\sqrt{E_s}$ in the complex plane, where $E_s$ denotes the symbol energy. Complex zero-mean Gaussian noise, with equal variance in the in-phase and quadrature components, is added to these points, yielding the received noisy symbols at the receiver. A nearest-neighbor detector is then employed to recover the transmitted message.

Except for the cases $M=2$ and $M=4$, the transition probability matrix of $M$-PSK, i.e., $\mathbf{T}$, does not admit a closed-form expression and must be computed numerically.

Let \(\delta=0\). As an example, for the message set \(\mathcal{M}=[1:8]\) and \(\mathbf{T}\) denoting the probability transition matrix of 8-PSK signaling over an AWGN channel, the curve of \(p_e^*(\mathbf{T})\) versus \(\epsilon\) is plotted in Fig.~\ref{fig_8PSK} for several SNR values, obtained by solving the corresponding LP. For comparison, the value of \(p_e^*\) in \eqref{pestar} (with \(M=8\) and \(\delta=0\)) is also shown in brown with triangular markers.

Furthermore, the upper bound in Proposition~1 is plotted using dashed lines. The lower bound \(l(\mathbf{T})\) is omitted from the figure, as it coincides with the horizontal asymptotes at which the curves saturate.

As can be observed, increasing the SNR brings the curves closer to \(p_e^*\), with the heuristic justification that, as the SNR increases, \(\mathbf{T}\) and \(\mathbf{P}^*_{\mathbf{T}}\) become closer to the identity matrix \(\mathbf{I}\) and \(\mathbf{P}^*\), respectively.

An important observation is that, in the high-privacy regime (i.e., small \(\epsilon\)), \(p_e^*(\mathbf{T})\) coincides with \(p_e^*\), implying that uncoded transmission is optimal in this regime. This is consistent with the third statement of theorem 2 due to the structure of $\mathbf{T}$, i.e., being doubly stochastic. Furthermore, one can derive bounds on the SNR required to achieve a given privacy level. For example, when \(\epsilon \leq 2\), there is no gain in reducing \(p_e^*(\mathbf{T})\) by increasing the SNR beyond \(10\,\mathrm{dB}\). In other words, when $p_e^*(\mathbf{T})$ meets $p_e^*$, the channel is good enough and further investment in reduction of the noise does not improve the worst-case error.

Finally, we observe similar behaviour at other values of \(\delta\); hence, the corresponding figures are omitted for brevity.

\begin{figure}[H]
 \centering 
 \scalebox{0.5} 
 {\includegraphics{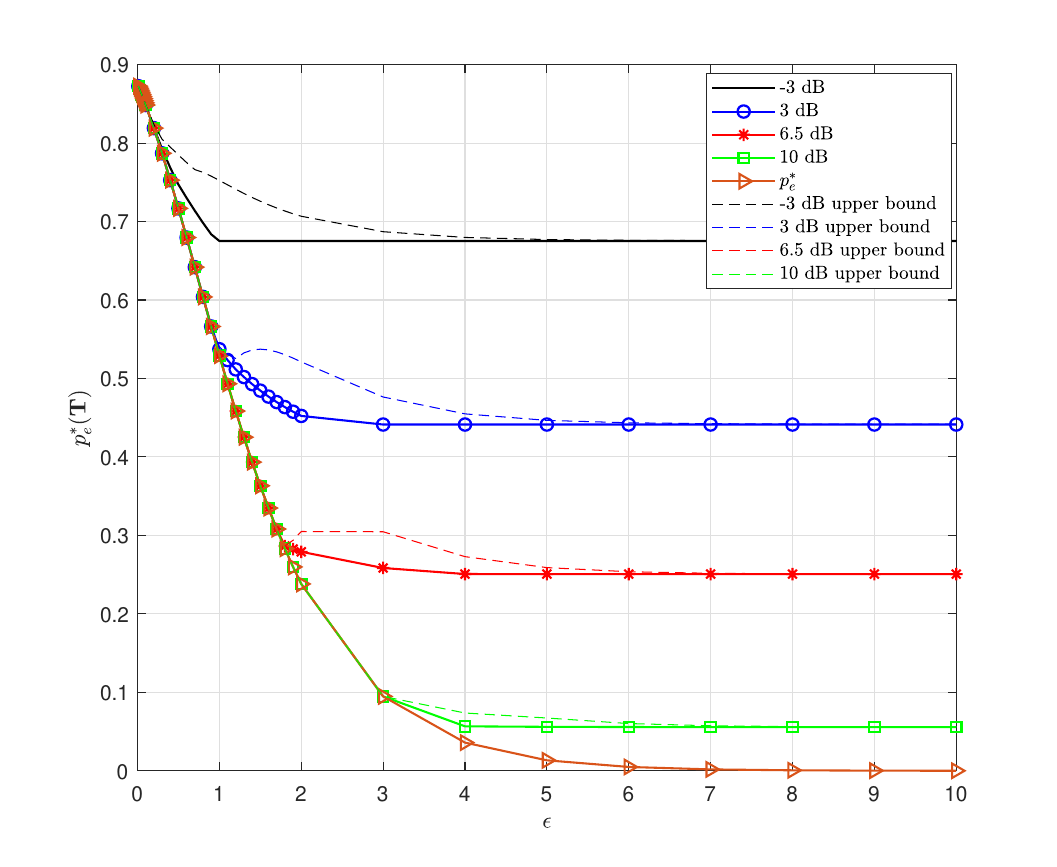}} 
 \caption{An example with $M=8,\delta=0$ and $\mathbf{T}$ capturing the probability transition of 8-PSK modulation in AWGN channel.}
 \label{fig_8PSK} 
\end{figure}

\section{Conclusion}

This paper studied the consistent release of counting queries under \((\epsilon,\delta)\)-differential privacy from a worst-case communication perspective. We derived a closed-form expression for the minimum achievable error probability and obtained an explicit canonical optimal mechanism. By exploiting the active privacy constraints, we further characterized the entire class of optimal mechanisms, identifying conditions for uniqueness, full rank, and the structural properties shared by all optimizers.

We also extended the framework to the setting where part of the communication medium is fixed. For this general model, we established necessary and sufficient conditions under which no utility loss occurs, derived upper and lower bounds on the optimal performance, and illustrated the results through (M)-ary PSK transmission over an AWGN channel, where uncoded transmission is shown to be effectively optimal in the high-privacy regime.

\appendices
\section{Proof of lemma \ref{lem1}}\label{App1}
Let $q\in (0,1)$, and define the matrix 
\(\mathbf{Z} \triangleq [q^{d_M(i,j)}].
\)
Since $z_{ij}$ depends only on $i-j \pmod M$, $\mathbf{Z}$ is a circulant matrix. As a result, we check its Fourier diagonalization.

The eigenvalues of this circulant matrix are
\[
\lambda_k = \sum_{r=0}^{M-1} q^{d_M(r,0)} e^{-2\pi i kr/M},
\]
and by symmetry, this becomes a real cosine sum.

\subsection*{Case 1: $M=2m+1$}

For $M=2m+1$, we have
\[
d(r,0)=
\begin{cases}
r, & 0 \le r \le m,\\
M-r, & m<r<M.
\end{cases}
\]
Hence
\[
\lambda_k
=
1 + 2 \sum_{r=1}^{m} q^r \cos\!\left(\frac{2\pi kr}{M}\right).
\]

Set $\theta \triangleq \frac{2\pi k}{M}$. Then
\[
\lambda_k
=
1 + 2 \sum_{r=1}^{m} q^r \cos(r\theta).
\]
Using the standard finite cosine sum identity\cite{grad},
\begin{align*}
&1 + 2 \sum_{r=1}^{m} q^r \cos(r\theta)
=\\
&\frac{
1 - q^2 - 2 q^{m+1} \cos((m+1)\theta)
+ 2 q^{m+2} \cos(m\theta)
}{
1 - 2q\cos\theta + q^2
}.
\end{align*}
Since $M=2m+1$, we have the identity
\[
\cos((m+1)\theta)=\cos(m\theta),
\]
so the numerator simplifies to
\[
1 - q^2 - 2 q^{m+1}(1-q)\cos(m\theta).
\]
Thus,
\[
\lambda_k
=
\frac{
(1-q)\bigl(1+q - 2 q^{m+1}\cos(m\theta)\bigr)
}{
1 - 2q\cos\theta + q^2
}.
\]
Since $| \cos(m\theta) | \le 1$ and $q^{m+1} < 1$, we obtain
\[
1+q - 2 q^{m+1}\cos(m\theta) > 0,
\]
and since the denominator is positive, $\lambda_k > 0$.

\subsection*{Case 2: $M=2m$}

We have
\[
\lambda_k
=
1 + 2 \sum_{r=1}^{m-1} q^r \cos(r\theta) + q^m (-1)^k,
\  \theta = \frac{\pi k}{m}.
\]

Using the finite cosine sum identity,
\begin{align*}
    &1 + 2 \sum_{r=1}^{m-1} q^r \cos(r\theta)
=\\
&\frac{
1 - q^2 - 2 q^m \cos(m\theta)
+ 2 q^{m+1} \cos((m-1)\theta)
}{
1 - 2q\cos\theta + q^2
}.
\end{align*}

Using \(\cos(m\theta)=(-1)^k,
\cos((m-1)\theta)=(-1)^k \cos\theta,
\)
we obtain
\[
\lambda_k
=
\frac{
1 - q^2 - (1-q)^2 q^m (-1)^k
}{
1 - 2q\cos\theta + q^2
}.
\]

Since the denominator is strictly positive and
\[
1 - q^2 - (1-q)^2 q^m (-1)^k > 0,
\]
it follows that $\lambda_k > 0$ for all $k$.

\section{}\label{App2}

\textbf{Generalization to arbitrary $M$.} Let \(F_i^k\) and \(L_i^k\) denote the sums of the first \(k\) and last \(k\) entries, respectively, of row \(i\) of \(\mathbf{P}^*\) in (\ref{optimizer}). Here, \(F\) and \(L\) stand for “first” and “last.”

For any $i\leq\lfloor\frac{M}{2}\rfloor$, we have
\begin{align}
    F_i^i&\triangleq\sum_{j=1}^ip^*_{i,j}\nonumber\nonumber\\
    & = \alpha\sum_{j=1}^i e^{-\epsilon d_M(i,j)} +\delta \nonumber\\
    &= \alpha e^{\epsilon}\sum_{j=1}^i e^{-\epsilon d_M(i+1,j)} +\delta\label{fact}\\
    & = e^\epsilon\sum_{j=1}^ip^*_{i+1,j}+\delta\nonumber\\
    &=e^\epsilon F_{i+1}^i+\delta,\nonumber
\end{align}
where the factorization of $e^\epsilon$ in (\ref{fact}) follows from the condition $1\leq j\leq i\leq\lfloor\frac{M}{2}\rfloor$. Due to the circulant symmetry of $\mathbf{P}^*$, we can do an analogous reasoning to $L_{M-i+1}^i$, and in conclusion, we obtain
\begin{align}
F_{i}^i&=e^\epsilon F_{i+1}^i+\delta\nonumber\\
   L_{M-i+1}^i&=e^\epsilon L_{M-i}^i+\delta,\ \ \ \ \ \forall i\leq\lfloor\frac{M}{2}\rfloor\label{Finc}
\end{align}
which are active constraints.

We follow the same reasoning as in the proof for the case \(M=5\). Suppose there exists \(\mathbf{P}' \in \mathcal{S}_{\epsilon,\delta}\) whose diagonal entries are strictly larger than those of \(\mathbf{P}^*\). Consider the evolution of \(\mathbf{P}^*\) toward \(\mathbf{P}'\). Along this direction, we observe an inevitable cascade of increments in the entries of \(\mathbf{P}^*\), as detailed below:

\begin{enumerate}
    \item $p_{i,i}^*$ increases $\forall i$.
    \item $F_i^i$ increases for $i\leq \lfloor\frac{M}{2}\rfloor.$ This follows from the previous step and induction, by noting that i) $F_{i+1}^{i+1} = F_{i+1}^i+p^*_{i+1,i+1}$ from definition, and ii)  $F_i^i\uparrow\Longrightarrow F_{i+1}^i\uparrow$ from (\ref{Finc}).
    \item Likewise, $L_{M-i+1}^i$ increases for $i\leq \lfloor\frac{M}{2}\rfloor.$
\end{enumerate}
These incremental effects force the sum of row \(\frac{M+1}{2}\) when \(M\) is odd, or the sums of rows \(\frac{M}{2}\) and \(\frac{M}{2}+1\) when \(M\) is even, to exceed one. Consequently, \(\mathbf{P}'\) cannot be a valid probability transition matrix, and hence \(\mathbf{P}^*\) is optimal. 

\textbf{Proof of the second part in statement 3.} Assume, to the contrary, that one of the diagonal entries in the upper-left half of an optimizer $\mathbf{P}$ satisfies $p_{i,i}>p^*_{i,i}$. (The argument for diagonal entries in the lower-right half is entirely analogous.) Also, assume \(M\) is odd, and the proof for even \(M\) is similar.

From the active constraints in \eqref{Finc}, we observe that $\sum_{j=1}^{i-1}p_{i,j}$ cannot be smaller than $F_{i-1}^i$, otherwise, moving upwards, we get \(p_{1,1}<p^*_{1,1}\), which contradicts optimality of \(\mathbf{P}\). Hence, we must have 
\begin{equation*}
    \sum_{j=1}^{i}p_{i,j}>F_i^i,
\end{equation*}
and propagating downwards, we get
\begin{equation*}
    \sum_{j=1}^{\frac{M+1}{2}}p_{\frac{M+1}{2},j}>F_{\frac{M+1}{2}}^{\frac{M+1}{2}}.
\end{equation*}
However, since each row sums to one, we must have
\begin{equation*}
    \sum_{j=\frac{M+1}{2}+1}^{M}p_{\frac{M+1}{2},j}<L_{\frac{M+1}{2}}^{\lfloor\frac{M}{2}\rfloor}.
\end{equation*}
Propagating this change downward through the active $\epsilon$-DP constraints in \eqref{Finc} then yields
\[
p_{M,M}<p^*_{M,M},
\]
contradicting the optimality. Therefore, no diagonal entry can exceed its counterpart in $\mathbf{P}^*$.

Finally, looking at the cascade of increments or decrements in this propagation argument, one can readily verify that for any $\mathbf{P}\in\mathcal{P}_{\textnormal{opt}}$,
\begin{align*}
    \sum_{j=1}^ip_{i,j}&=\sum_{j=1}^ip^*_{i,j},\ \ i\leq\lceil\frac{M+1}{2}\rceil\nonumber\\
    \sum_{j=i}^Mp_{i,j}&=\sum_{j=i}^Mp^*_{i,j},\ \ i\geq\lceil\frac{M+1}{2}\rceil.
\end{align*}
In each case, the assumption to the contrary leads to either \(p_{1,1}<p^*_{1,1}\) or \(p_{M,M}<p^*_{M,M}\), which contradict the optimality of $\mathbf{P}$.

\section{Proof of lemma \ref{lem2}}\label{App3}
Obviously, we have
\begin{align}
\Delta_M(\epsilon,\delta)
&=
\sup_{\mathbf{P}\in\mathcal{P}_{\mathrm{opt}}}
\|\mathbf{P}-\mathbf{P}^*\|_{\max}\nonumber\\
&\leq \frac{1}{2}
\sup_{\mathbf{P}\in\mathcal{P}_{\mathrm{opt}}}
\|\mathbf{P}-\mathbf{P}^*\|_{\infty}\label{sum0}\\
&\leq (1-\delta)\bigl(1-\alpha_M(\epsilon)\bigr),\label{mosht}
\end{align}
where \eqref{sum0} follows from the fact that every row of
$\mathbf{P}-\mathbf{P}^*$ sums to zero, while \eqref{mosht} follows from
Theorem~1 (statement 3), which states that all optimal mechanisms share the same
diagonal entries as $\mathbf{P}^*$.

Although the bound in \eqref{mosht} admits a simple closed-form expression,
it is rather loose in the high-privacy regime, namely as
$\epsilon,\delta\to0$. Indeed, substituting $(\epsilon,\delta)=(0,0)$ into
\eqref{mosht} yields
\[
\Delta_M(0,0)\leq\frac{M-1}{M},
\]
whereas, by the uniqueness of the optimal solution,
\[
\Delta_M(0,0)=0.
\]
To address this concern, for any \(\mathbf{P}\in\mathcal{P}_{\textnormal{opt}}\), define
\begin{align*}
    \Phi(\mathbf{P})&\triangleq\max_{i,j: i\neq j} p_{i,j}\\
    \phi(\mathbf{P})&\triangleq\min_{i,j: i\neq j} p_{i,j}
\end{align*}
Since \(\mathbf{P,P^*}\) share the same diagonal entries, we have
\begin{align*}
    \|\mathbf{P}-\mathbf{P}^*\|_{\textnormal{max}}\leq\max\{\Phi(\mathbf{P})-\phi(\mathbf{P}^*),\Phi(\mathbf{P}^*)-\phi(\mathbf{P})\}.
\end{align*}
We already know that
\begin{align*}
    \Phi(\mathbf{P}^*)&=(1-\delta)\alpha_M(\epsilon)e^{-\epsilon}\\
    \phi(\mathbf{P}^*)&=(1-\delta)\alpha_M(\epsilon)e^{-\lfloor\frac M2\rfloor\epsilon}
\end{align*}
as the maximum and minimum off-diagonal entries in \(\mathbf{P^*}\). It remains to find an upper bound and a lower bound for $\Phi(\mathbf{P})$ and $\phi(\mathbf{P})$, respectively.

We use a simple bound based on the known diagonal entries in \(\mathbf{P}\) and consecutive application of DP constraints.

As a result, for any optimizer, we get
\begin{align*}
    \Phi(\mathbf{P})&\leq\Phi_0\triangleq\min\{\overline{\Phi},(1-\alpha_M(\epsilon))(1-\delta)\}\\
    \phi(\mathbf{P})&\geq\phi_0\triangleq\max\{\underline{\phi},0\}
\end{align*},
with 

\begin{align*}
    \overline{\Phi}&\triangleq e^{(M-1)\epsilon}(\delta+(1-\delta)\alpha_M(\epsilon))+\delta\frac{1-e^{(M-1)\epsilon}}{1-e^\epsilon}\\
    \underline{\phi}&\triangleq\frac{\delta+(1-\delta)\alpha_M(\epsilon)-\delta\frac{1-e^{(M-1)\epsilon}}{1-e^\epsilon}}{e^{(M-1)\epsilon}}.
\end{align*}
Hence, we have
\begin{align*}
   \sup_{\mathbf{P}\in\mathcal{P}_{\mathrm{opt}}} \|\mathbf{P}-\mathbf{P}^*\|_{\textnormal{max}}&\leq\max\{\Phi_0-\phi(\mathbf{P}^*),\Phi(\mathbf{P}^*)-\phi_0\}\\
   &\triangleq R_M(\epsilon,\delta).
\end{align*}
Finally, it can be verified that when \(\epsilon,\delta\to 0\), we have \(R_M(\epsilon,\delta)\to 0\).

\section{}\label{App4}
We have
\begin{equation*}
    f(u,v)\triangleq (k+\delta')^3-(k+\delta')\bigl(uv+4-(u+v)\bigr)-2\bigl(uv-(u+v)\bigr).
\end{equation*}
Since \(u,v\in[\frac{1-\delta'}{k},1]\), we can write
\begin{align*}
    \frac{\partial f}{\partial u}
&=(k+\delta'+2)(1-v)\geq 0\\
\frac{\partial f}{\partial v}
&=(k+\delta'+2)(1-u)\geq0.
\end{align*}

Hence, \(f(u,v)\) is nondecreasing in both \(u\) and \(v\), and therefore
\[
f(u,v)
\geq g(k,\delta')\triangleq
f\!\left(\frac{1-\delta'}{k},
\frac{1-\delta'}{k}\right).
\]
It remains to prove that $g(k,\delta')>0$.

We have
\begin{align*}
g(k,\delta')&=
(k+\delta')^3
-4(k+\delta')
+\frac{2(k+\delta'+2)(1-\delta')}{k}\nonumber\\
&\ \ \ -\frac{(k+\delta'+2)(1-\delta')^2}{k^2}\\
&= (k+\delta'+2)(1-\frac{1}{k^2})(k+\delta'-1)^2\\
&>0,
\end{align*}
since $k>1$.

Therefore, \(f(u,v)\geq g(k,\delta')>0,
\)
which completes the proof.

\bibliography{REF}
\bibliographystyle{IEEEtran}
 \end{document}